\documentclass[letter,11pt]{article}

\usepackage{amsmath,amsfonts,amssymb,amsthm}
\usepackage{amsthm}
\usepackage{graphicx,color}
\usepackage{framed}
\usepackage{enumitem}
\usepackage{thmtools}
\usepackage{thm-restate}
\usepackage{xspace}
\usepackage{bm}
\usepackage{footnote}
\usepackage{mathtools}
\usepackage{mathrsfs}
\usepackage[disable]{todonotes}
\usepackage[most]{tcolorbox}
\usepackage{footnote}
\usepackage{rotating}
\usepackage[pdftex, plainpages = false, pdfpagelabels, 
bookmarks=true,
bookmarksopen = true,
bookmarksnumbered = true,
breaklinks = true,
linktocpage,
pagebackref,
colorlinks = true,  
linkcolor = blue,
urlcolor  = blue,
citecolor = red,
anchorcolor = green,
hyperindex = true,
hyperfigures
]{hyperref} 
\usepackage{tabularx}
\usepackage{tikz} 
\usetikzlibrary{calc}
\usepackage{thm-autoref}
\usepackage[nameinlink]{cleveref}

\usepackage{algorithm}
\usepackage{algorithmicx}
\usepackage{algpseudocode}

\algtext*{EndWhile}
\algtext*{EndIf}
\algtext*{EndFor}
\algtext*{EndProcedure}
\algtext*{EndFunction}

\newtheorem{lemma}{Lemma}
\numberwithin{lemma}{section}
\newtheorem{theorem}[lemma]{Theorem}
\newtheorem{corollary}[lemma]{Corollary}
\newtheorem{definition}[lemma]{Definition}
\newtheorem{observation}[lemma]{Observation}

\newtheorem{fact}[lemma]{Fact}

\theoremstyle{definition}

\theoremstyle{definition}

\usepackage{geometry}
\let\oldlambda\lambda
\renewcommand{\lambda}{\ensuremath{\oldlambda}\xspace}
\let\oldalpha\alpha
\renewcommand{\alpha}{\ensuremath{\oldalpha}\xspace}
\let\oldDelta\Delta
\renewcommand{\Delta}{\ensuremath{\oldDelta}\xspace}

\let\mydelta\delta
\renewcommand{\delta}{\ensuremath{\mydelta}\xspace}

\let\mytau\tau
\renewcommand{\tau}{\ensuremath{\mytau}\xspace}

\let\mytheta\theta
\renewcommand{\theta}{\ensuremath{\mytheta}\xspace}

\let\mygamma\gamma
\renewcommand{\gamma}{\ensuremath{\mygamma}\xspace}

\let\myGamma\Gamma
\renewcommand{\Gamma}{\ensuremath{\myGamma}\xspace}
\usepackage{tikz}
\usetikzlibrary{arrows,positioning}

\newcommand{\cO}{\ensuremath{\mathcal{O}}\xspace}

\newcommand{\bigoh}{\mathcal{O}}

\newcommand{\tw}{\operatorname{tw}}
\newcommand{\fhtw}{\operatorname{fhtw}}
\newcommand{\htw}{\operatorname{htw}}
\newcommand{\ghtw}{\operatorname{ghtw}}
\newcommand{\mis}
{\alpha}
\newcommand{\cov}{\rho_H^*}
\newcommand{\cost}{\operatorname{cost}}

\newcommand{\width}{w}

\DeclarePairedDelimiterX{\norm}[1]{\lVert}{\rVert}{#1}
\newcommand{\gf}[1]{\underline{#1}}

\newcommand{\treeT}{T}

\newcommand{\hy}{\hat{y}}
\newcommand{\hz}{\hat{z}}

\newcommand{\optab}{{\sf opt}_{ab}}
\newcommand{\polyH}{$\norm{H}^{\cO(1)}$}
\newcommand{\cB}{\mathcal{B}}

\newcommand{\cP}{\mathcal{P}}
\newcommand{\cF}{\mathcal{F}}

\newcommand{\supp}{\operatorname{supp}}
\newcommand{\hS}{\hat{S}}

\title{Efficient Approximation of Fractional Hypertree Width}

\author{
    Viktoriia Korchemna\thanks{TU Wien, Vienna, Austria.}
     \and
	Daniel Lokshtanov\thanks{University of California Santa Barbara, USA.}
	\and
	Saket Saurabh \thanks{
Department of Informatics, University of Bergen, Norway} \thanks{Institute of Mathematical Sciences, Chennai, India.}
	\and
	Vaishali Surianarayanan\footnotemark[2]
    \and
    Jie Xue\thanks{New York University Shanghai, China.}
 }

\date{}

\begin{document}

\maketitle

\begin{abstract}
We give two new approximation algorithms to compute the {\em fractional hypertree width} of an input hypergraph. The first algorithm takes as input $n$-vertex $m$-edge hypergraph $H$ of fractional hypertree width at most $\omega$, runs in polynomial time and produces a tree decomposition of $H$ of fractional hypertree width $\cO(\omega \log n \log \omega)$, i.e., it is an $\cO(\log n \log \omega)$-approximation algorithm. As an immediate corollary this yields polynomial time $\cO(\log^2 n \log \omega)$-approximation algorithms for (generalized) hypertree width as well. To the best of our knowledge our algorithm is the first non-trivial polynomial-time approximation algorithm for fractional hypertree width and (generalized) hypertree width, as opposed to algorithms that run in polynomial time only when $\omega$ is considered a constant. For hypergraphs with the {\em bounded intersection property} (i.e. hypergraphs where every pair of hyperedges have at most $\eta$ vertices in common) the algorithm outputs a hypertree decomposition with fractional hypertree width $\cO(\eta \omega^2 \log \omega)$ and generalized  hypertree width $\cO(\eta \omega^2 \log \omega (\log \eta + \log \omega))$. This ratio is comparable with the recent algorithm of Lanzinger and Razgon [STACS 2024], which produces a hypertree decomposition with generalized hypertree width $\cO(\omega^2(\omega + \eta))$, but uses time (at least) exponential in $\eta$ and $\omega$.

The second algorithm runs in time $n^{\omega}m^{\cO(1)}$ and produces a tree decomposition of $H$ of fractional hypertree width $\cO(\omega \log^2 \omega)$. This significantly improves over the $(n+m)^{\cO(\omega^3)}$ time algorithm of Marx [ACM TALG 2010], which produces a tree decomposition of fractional hypertree width $\cO(\omega^3)$, both in terms of running time and the approximation ratio.

Our main technical contribution, and the key insight behind both algorithms, is a variant of the classic Menger's Theorem for clique separators in graphs: For every graph $G$, vertex sets $A$ and $B$, family ${\cal F}$ of cliques in $G$, and positive rational $f$, either there exists a sub-family of $\cO(f \cdot \log^2 n)$ cliques in  ${\cal F}$ whose union separates $A$ from $B$, or there exist $f \cdot \log |{\cal F}|$ paths from $A$ to $B$ such that no clique in ${\cal F}$ intersects more than $\log |{\cal F}|$ paths.








\end{abstract}

\newpage
\setcounter{page}{1}
\section{Introduction}
A hypergraph $H$ consists of a set $V(H)$ of vertices and a family $E(H)$ of hyperedges, where each hyperedge $e \in E(H)$ is a subset of $V(H)$. A {\em tree decomposition} of a hypergraph $H$ is a pair $(T, \beta)$ where $T$ is a tree and $\beta : V(T) \rightarrow 2^{V(H)}$ is a function that assigns to each vertex $u$ of $T$ a set $\beta(u)$ (called a {\em bag}) of vertices of $H$. To be a tree decomposition of $H$ the pair $(T, \beta)$ must satisfy the following two axioms:  {\em (i)} For every $v \in V(H)$ the set $\{u \in V(T) : v \in \beta(u)\}$ induces a non-empty and connected sub-tree of $T$, and {\em (ii)} For every $e \in E(H)$ there exists a $u \in V(T)$ such that $e \subseteq \beta(u)$.

For many algorithmic problems the structure of the input instance is best expressed as a hypergraph. For such problems a suitable tree decomposition of the input hypergraph can often be exploited algorithmically. Algorithms based on hypergraph tree decompositions have been designed for problems that are motivated by different application areas, such as databases~\cite{DBLP:journals/jcss/GottlobLS02,DBLP:journals/tcs/GottlobLPR23,DBLP:journals/jacm/GottlobLPR21}, constraint satisfaction problems~\cite{DBLP:journals/talg/Marx10,DBLP:journals/talg/GroheM14,DBLP:journals/jacm/GottlobLPR21}, combinatorial auctions~\cite{DBLP:journals/jacm/GottlobG13}, and automated selection of Web services based on recommendations from social networks~\cite{DBLP:journals/isci/HashmiMNR16}. 
%
Hypergraph tree decomposition based algorithms have been particularly impactful in databases -- being applied both in commercial database systems such as LogicBlox~\cite{DBLP:conf/sigmod/ArefCGKOPVW15,DBLP:conf/pods/KhamisNRR15,DBLP:conf/pods/KhamisNR16,DBLP:conf/sigmod/TuR15}  and advanced research
prototypes such as EmptyHeaded~\cite{DBLP:conf/sigmod/AbergerTOR16,DBLP:conf/icde/AbergerTOR16,DBLP:journals/pvldb/AbergerLOR17}. 
For this reason hypergraph decompositions have been extensively studied both from the perspective of algorithms with provable guarantees in terms of time~\cite{DBLP:journals/jcss/GottlobLS02,DBLP:journals/tcs/GottlobLPR23,DBLP:journals/talg/Marx10,DBLP:journals/talg/GroheM14,DBLP:journals/jacm/GottlobLPR21,DBLP:conf/stacs/LanzingerR24,DBLP:journals/ipl/MollTT12}, space~\cite{DBLP:journals/jcss/GottlobLS02} and parallellism~\cite{DBLP:journals/tods/GottlobLOP24,DBLP:journals/jcss/GottlobLS02,DBLP:conf/icdt/AfratiJRSU17}, and from the perspective of experimental algorithmics~\cite{DBLP:journals/jea/FischlGLP21,DBLP:journals/jea/GottlobLLO22}. 

For more details on how hypergraph decompositions can be exploited in the constraint satisfaction problem, see \cite{DBLP:journals/talg/GroheM14}. Here, one usually considers a \emph{constraint hypergraph}, whose vertices are
the variables and whose hyperedges are the sets of all those variables which
occur together in a constraint scope. It is hard to overestimate the importance of the constraint satisfaction problem, since many other problems from different areas of Computer Science can be reduced or linked to it. One example here is the
game theory problem of computing pure Nash equilibria in graphical games. Namely, bounded hypertree width of the hypergraph representing the dependency structure of a game along with some other constraints allows to find pure Nash and Pareto equilibria in polynomial time~\cite{GottlobGS05}. Another example is the homomorphism problem, where the question is whether there exists a hyperedge preserving mapping from one hypergraph to another~\cite{GottlobLS01}.  

Algorithms that use hypergraph tree decompositions need the decomposition to have low {\em width}, for an appropriate notion of width. There are several related width measures for hypergraph decompositions, the most commonly used ones to date are hypertree width~\cite{DBLP:journals/jcss/GottlobLS02}, generalized hypertree width~\cite{DBLP:journals/jcss/GottlobLS02}, and fractional hypertree width~\cite{DBLP:journals/talg/GroheM14}.
For instance, bounded fractional hypertree width allows to solve the constraint satisfaction problem in polynomial time, provided that the tree decomposition is given in the input~\cite{DBLP:journals/talg/GroheM14}.
We now define generalized hypertree width, and fractional hypertree width. We will not formally define hypertree width in this paper, for a formal definition see e.g.~\cite{DBLP:journals/jcss/GottlobLS02}.

For a vertex set $S$ and a family ${\cal F}$ of vertex sets we say that the {\em cover} number $\rho_{\cal F}(S)$ of $S$ with respect to ${\cal F}$ is the minimum size of a sub-family ${\cal F}'$ such that the union of the sets in ${\cal F}'$ contains $S$. The {\em fractional cover} number $\rho_{\cal F}^*(S)$ of $S$ with respect to ${\cal F}$ is the minimum $\sum_{f \in {\cal F}} \gamma(f)$ taken over all functions $\gamma : {\cal F} \rightarrow [0, 1]$ such that for every $v \in S$ it holds that 
$$\sum_{\substack{f \in {\cal F} \mbox{ s.t. } \\ v \in f}} \gamma(f) \geq 1$$
The {\em generalized hypertree width} of a tree decomposition $(T, \beta)$ of a hypergraph $H$ is the maximum cover number of a bag with respect to $E(H)$. Similarly, the {\em fractional hypertree width} of $(T, \beta)$ is the maximum fractional cover number of a bag with respect to $E(H)$. The generalized (fractional) hypertree width $\ghtw(H)$ ($\fhtw(H)$) of a hypergraph $H$ is the minimum generalized (fractional) hypertree width of a tree decomposition $(T, \beta)$ of $H$. The {\em hypertree width} of a hypergraph $H$, which we will not explicitly define here, is denoted by $\htw(H)$. The different width notions are quite closely related to each other, in particular it is known (see~\cite{DBLP:journals/ejc/AdlerGG07} and~\cite{DBLP:journals/jacm/GottlobLPR21} for all inequalities except for the last one, and~\cite{williamson2011design} for the last one) that for every hypergraph $H$,
\begin{align}\label{eqn:widthRelations}
\fhtw(H) \leq \ghtw(H) \leq \htw(H) \leq 3 \cdot \ghtw(H) \leq \cO(\log n \cdot \fhtw(H))
\end{align}
    
Despite the significant success of hypergraph decomposition methods, more widespread adoption is being held back by how computationally expensive it currently is to compute tree decompositions of bounded width~\cite{DBLP:journals/tods/GottlobLOP24}. 
%
%
We now summarize what is known, focusing on exact computation and approximation algorithms whose running time is polynomial when the appropriate width measure of the input hypergraph is assumed to be a fixed constant independent of input size. 
There exists a linear time algorithm to determine whether an input hypergraph $H$ has generalized hypertree width $1$. Indeed, such hypergraphs are precisely the {\em acyclic} hypergraphs (see~\cite{DBLP:journals/jacm/GottlobLPR21}), and a linear time algorithm for determining whether a hypergraph is acyclic was given by Yannakakis~\cite{DBLP:conf/vldb/Yannakakis81}. 
On the other hand it is NP-hard~\cite{DBLP:journals/jacm/GottlobLPR21} to determine whether an input hypergraph has generalized hypertree width at most $\omega$ for every $\omega \geq 2$. The same hardness result holds true for fractional hypertree width~\cite{DBLP:journals/jacm/GottlobLPR21}. For hypertree width the situation is better. In particular there exists an algorithm with running time $n^{\cO(\omega)}$ to determine whether the input hypergraph has hypertree width at most $\omega$~\cite{DBLP:journals/jcss/GottlobLS02}. Together with (\ref{eqn:widthRelations}) this yields a constant factor approximation for generalized hypertree width with running time $n^{\cO(\omega)}$. Marx~\cite{DBLP:journals/talg/Marx10} gave an approximation algorithm for fractional hypertree width, which given a hypergraph $H$ with $m$ hyperedges and fractional hypertree width at most $\omega$, produces a tree decomposition of $H$ of width at most $\cO(\omega^3)$ in time $(n+m)^{\cO(\omega^3)}$.

In the running time of all of the above algorithms the exponent of $n$ grows with $\omega$. For hypertree width and generalized hypertree width, this is essentially unavoidable:
It is known that hypertree width and generalized hypertree width are at least as hard (both to compute exactly, and to approximate~\cite{DBLP:conf/wg/GottlobGMSS05}) as the {\sc Set Cover}  problem~\cite{DBLP:journals/siamcomp/DowneyF95}. Hence, from the strong inapproximability results for {\sc Set Cover}~\cite{DBLP:journals/jacm/SLM19} it follows that assuming the Exponential Time Hypothesis (ETH) there does not exist $f(\omega)n^{o(\omega)}$ time $g(\omega)$-approximation algorithms for hypertree width or generalized hypertree width for any functions $f$, $g$. On the other hand, for fractional hypertree width, the strongest known inapproximability result is that it is at least as hard to approximate as the treewidth of an input graph. In particular it is hard to constant factor approximate (for every constant) in $n^{\cO(1)}$ time assuming the Small Set Expansion Hypothesis~\cite{DBLP:journals/jair/WuAPL14}. 

\paragraph{Our Results.} Previous to our work, absolutely nothing was known about ``truly'' polynomial time (approximation) algorithms for (fractional / generalized) hypertree width, that is, the algorithms with no super-polynomial dependency on the value of the optimal solution in the running times. The idea that no such algorithms with a meaningful approximation ratio can exist is so pervasive that most discussions of ``tractability'' or ``polynomial time'' in the context of hypergraph decompositions silently assumes that the width of the input hypergraph is constant (see e.g. Gottlob et al.~\cite{DBLP:journals/jacm/GottlobLPR21}). Our first main result is the first non-trivial polynomial time approximation algorithm for fractional hypertree width (and therefore, by (\ref{eqn:widthRelations}), for generalized hypertree width and hypertree width as well). 


\begin{restatable}{theorem}{MainPolyTimeApprox}
\label{theorem:main-poly-approx}  
There exists an algorithm that given a hypergraph $H$, and a real number $\omega$, runs in time $\norm{H}^{\cO(1)}$ and either concludes that $\fhtw(H)>\omega$ or constructs a tree decomposition of $H$ with fractional hypertree width at most $\cO(\omega \log\omega\min\{\ln\alpha(H),\mu_H,\omega \cdot \eta_H\})$, where $\alpha(H)$ is the size of the maximum independent set in $H$, $\mu_H$ is the degeneracy of the bipartite incidence graph of $H$, and $\eta_H$ is the maximum size of the intersection of two hyperedges of $H$.
\end{restatable} 

The definition of degeneracy is provided in Section~\ref{sec:prelims}. In particular, the algorithm of Theorem~\ref{theorem:main-poly-approx} yields a $\cO(\log \omega \log n)$-approximation for fractional hypertree width in the worst case, and achieves even better approximation guarantees for several interesting restricted classes of hypergraphs. 

A restricted case that deserves special attention is when the size of the maximum intersection $\eta_H$ of two hyperedges in $H$ is upper bounded by a constant. It turns out that hypergraphs that arise in practical applications quite often have this property~\cite{DBLP:journals/jea/FischlGLP21}. This case was very recently considered by Lanzinger and Razgon~\cite{DBLP:conf/stacs/LanzingerR24}. They gave an algorithm that takes as input a hypergraph $H$ of generalized hypertree width $\omega$, runs in time $f(\omega, \eta_H)n^{\cO(1)}$, and produces a tree decomposition of $H$ of generalized hypertree width $\cO(\omega^2(\omega + \eta_H))$. The running time dependence $f$ on $\omega$ and $\eta_H$ of their algorithm is at least exponential. Lanzinger and Razgon~\cite{DBLP:conf/stacs/LanzingerR24} pose as an open problem whether one can get an algorithm with similar running time that produces a tree decomposition of $H$ with sub-cubic generalized hypertree width.

In contrast to the exponential algorithm for {\em generalised} hypertree width mentioned above, the algorithm of Theorem~\ref{theorem:main-poly-approx} uses polynomial time to output a tree decomposition of $H$ of {\em fractional} hypertree width at most $\cO(\omega^2 \eta_H \log \omega)$. It turns out that for hypergraphs with small $\eta_H$, every set with fractional cover number at most $k$ has cover number at most $\cO(k(\log k + \log \eta_H))$ (see Lemma~\ref{lem: bip_fractional_to_integral}). Thus, the tree decomposition output by Theorem~\ref{theorem:main-poly-approx} has generalized hypertree width at most $\cO(\omega^2 \eta_H\log\omega (\log \omega + \log \eta_H))$.

Thus, the algorithm of Theorem~\ref{theorem:main-poly-approx} achieves in polynomial time an approximation ratio that is comparable (up to factors polylogarithmic in $\log \omega$ and $\log \eta_H$) to the algorithm of Lanzinger and Razgon~\cite{DBLP:conf/stacs/LanzingerR24}. When $\eta_H$ is much smaller than $\omega$ (i.e if $\eta_H = \cO(\omega^{0.99})$) the generalized hypertree width of the tree decomposition output by the algorithm of Theorem~\ref{theorem:main-poly-approx} is sub-cubic in $\omega$, partially resolving the open problem of  Lanzinger and Razgon~\cite{DBLP:conf/stacs/LanzingerR24} in the affirmative. 

The linear programming based approach of Theorem~\ref{theorem:main-poly-approx} turns out to also be very helpful for designing approximation algorithms in the setting where the width $\omega$ of the input hypergraph is considered to be constant. Indeed, by adding an $n^{\cO(\omega)}$ time iterative improvement scheme on top of the algorithm of Theorem~\ref{theorem:main-poly-approx}, and exploiting that the approximation ratio of that algorithm is bounded by $\cO(\log\omega\log\alpha(H))$, we obtain an improved approximation algorithm for fractional hypertreewidth. 

\begin{restatable}{theorem}{MainLgApprox}   
\label{theorem:main-FPT-approx}
There exists an algorithm that given a hypergraph $H$ with $n$ vertices, and a rational number $\omega$, runs in time $n^{\lfloor\omega\rfloor}\norm{H}^{\cO(1)}$ and either constructs a tree decomposition of $H$ with fractional hypertree width at most $\cO(\omega\log^2\omega)$ or concludes $\fhtw(H)>\omega$.
\end{restatable} 

Our algorithm substantially improves over the previously best known approximation algorithm by Marx~\cite{DBLP:journals/talg/Marx10} both in the approximation ratio (from $\cO(\omega^2)$ to  $\cO(\log^2 \omega)$) and in the running time (from $n^{\cO(\omega^3)}$ to $n^{\cO(\omega)}$).

\medskip
\noindent {\bf Methods.} Our main technical contribution is a lemma that lets us efficiently find separators that are (fractionally) covered by few hyperedges. More precisely, let $H$ be a hypergraph, and $A$ and $B$ be vertex sets in $H$. A function $x : V(G) \rightarrow [0, 1]$ is a {\em fractional $A$-$B$ separator} if, for every $A$-$B$ path $P$ in $H$ we have that $\sum_{v \in V(P)} x(v) \geq 1$. A vertex set $X$ is an $A$-$B$ {\em separator} if every $A$-$B$ path intersects $X$.
%
Let now ${\cal F}$ be a family of vertex sets in $H$. We can generalize the notion of fractional covers (of sets) to also apply to fractional separators in a natural way. Namely, the {\em fractional} ${\cal F}$-\emph{edge cover} number $\rho_{\cal F}^*(x)$ of $x$ is the minimum $\sum_{f \in {\cal F}} \gamma(f)$ taken over all functions $\gamma : {\cal F} \rightarrow [0, 1]$ such that for every $v \in V(G)$ it holds that 
$$\sum_{\substack{f \in {\cal F} \mbox{ s.t. } \\ v \in f}} \gamma(f) \geq x(v)$$
Observe that when $x(v) \in \{0,1\}$ and $x$ corresponds to a vertex set $X$ then the fractional cover number of $x$ is the fractional cover number of $X$. We are now ready to state our main lemma.

\begin{restatable}{lemma}{mainCliqueSep}
\label{lem:mainCliqueSepSimplified}
Let $H$ be a hypergraph, $A$ and $B$ be vertex sets in $H$ and $x$ be a fractional $A$-$B$ separator in $G$. Then there exists an $A$-$B$ separator $X$ such that $\rho_{E(H)}^*(X) \leq \rho_{E(H)}^*(x) \cdot \cO(\log n)$.
\end{restatable}

Lemma~\ref{lem:mainCliqueSepSimplified} is stated and proved, in a slightly strengthened form, as Theorem~\ref{theorem:ab-sep-stronger}. Lemma~\ref{lem:mainCliqueSepSimplified} allows us to transform {\em fractional} separators that are {\em fractionally} covered by few cliques to {\em normal} separators that are {\em fractionally} covered by few cliques, at the cost of a $\cO(\log n)$ factor in the cover number. 
Since $\rho_{\cal F}(X) \leq  \rho_{\cal F}^*(X) \cdot \cO(\log n)$ for every vertex set $X$ (see e.g.~\cite{williamson2011design}), Lemma~\ref{lem:mainCliqueSepSimplified} can also be used to transform {\em fractional} separators that are {\em fractionally} covered by few cliques to {\em normal} separators that are covered by few cliques, at the cost of a $\cO(\log^2 n)$ factor in the cover number. Combined with standard sampling tecniques, Lemma~\ref{lem:mainCliqueSepSimplified} implies the following approximate version of the classic Menger's Theorem~\cite{menger1927allgemeinen} for clique separators.

\begin{restatable}[Clique Menger's Theorem]{theorem}{mainMenger}
\label{thm:mainMengerSimplified}
Let $G$ be a graph on $n$ vertices, $A$ and $B$ be vertex sets in $G$, ${\cal F}$ be a set of cliques in $G$, and $f$ be a positive real. 
Then, either there exists an $(A,B)$-separator $X$ such that $\rho_{G,\cal F}^*(X) \leq (8+4\ln n)f$ or there exists a collection $P_1, P_2, \ldots, P_{\lceil f \cdot \log(|{\cal F}|) \rceil}$ of $A$-$B$ paths such that no clique in ${\cal F}$ intersects at least $6\log(|{\cal F}|)$ of the paths.
\end{restatable}

We expect Lemma~\ref{lem:mainCliqueSepSimplified} and Theorem~\ref{thm:mainMengerSimplified} to find further applications both within and outside of the context of hypergraph width parameters. Indeed, Theorem~\ref{thm:mainMengerSimplified} has already been used as a crucial ingredient for the first quasi-polynomial time algorithm for the {\sc Independent Set} problem on even-hole-free graphs~\cite{ChudGHLS24}. 
    
\paragraph{Overview of the paper.}
In Section~\ref{sec:overview}, we give a brief technical overview of our results. Then in Section~\ref{sec:prelims}, we provide the required notations and preliminaries for our paper. In the subsequent Section~\ref{sect:ab-sep}, we give our main rounding algorithm for $(A,B)$-separators in Theorem~\ref{theorem:ab-sep-main} and Theorem~\ref{theorem:ab-sep-stronger}. Then we turn to our rounding algorithm for computing balanced separators in Theorem~\ref{theorem:bal-separator-main} in Section~\ref{section:bal-sep}. Equipped with our approximation algorithm for balanced separator from Theorem~\ref{theorem:bal-separator-main}, in Section~\ref{sec:fhtw-algorithm} we prove our two main approximation algorithms for $\fhtw$ in Theorem~\ref{theorem:main-poly-approx} and Theorem~\ref{theorem:main-FPT-approx}. In this section we also prove some nice properties of graphs with bounded size edge intersections. Then in the Section~\ref{sec:clique-menger}, we provide our clique version of Menger's Theorem (which we call Clique Menger's Theorem) in Theorem~\ref{thm:mainMengerSimplified}. In the same section, we also show how to compute fractional $(A,B)$-separator and provide a $\log n$ gap instance for the LP. Finally, in Section~\ref{sec:conclusion}, we conclude and provide some interesting open problems. 
\section{Brief Overview of Results}
\label{sec:overview}
In this section we give the intuition for the proofs of our main results. In this overview we will focus on a weaker variant of Theorem~\ref{theorem:main-poly-approx} where the fractional hypertree width of the output decomposition is at most $\cO(\omega \cdot \log \omega \cdot \log n)$. We will then remark how to obtain (some of) the improved bounds. 

The overall scheme of Theorem~\ref{theorem:main-poly-approx} follows the beaten path of approximation algorithms for treewidth~\cite{leighton1999multicommodity,FeigeHL08,DBLP:conf/uai/Amir01,DBLP:journals/jal/BodlaenderGHK95, DBLP:journals/jal/BodlaenderGHK95} and related problems (including the approximation algorithm for fractional hypertreewidth by Marx~\cite{DBLP:journals/talg/Marx10}). In particular we apply the scheme first used by Robertson and Seymour~\cite{RobertsonS95b} for approximating the treewidth. Further our algorithm also goes along the lines of the $\log n$ approximation for tree width by Bodlaender et.al.~\cite{DBLP:journals/jal/BodlaenderGHK95} that  combines this scheme with a ball growing argument similar to that of Leighton and Rao  \cite{leighton1999multicommodity} for finding a balanced separator.

In our scheme, we have as input a hypergraph $H$ and a subset $Z$ of vertices. The task is to find a tree decomposition of $H$ of fractional hypertree width at most $\cO(\omega \cdot \log \omega \cdot \log n)$ with $Z$ as the root bag, or to conclude that the fractional hypertree width of $H$ is greater than $\omega$. Standard results on tree decompositions (see e.g.~\cite{cygan2015parameterized}) imply that assuming that the set $Z$ is ``large enough'' compared to the fractional hypertree width $\omega$ of $H$, there exists a ``small'' and ``balanced'' separator $S$ for $Z$ in $G$. The precise meaning of ``large enough'', ``small'' and ``balanced'' depend on the problem in question. In our setting, $Z$ is ``large enough'' means that the fractional cover number (with respect to $E(H)$) of $Z$ needs to be at least $2\omega$, ``small'' for $S$ means that the fractional cover number of $S$ is at most $\omega$ and ``balanced'' means that for every connected component $C$ of $H - S$ it holds that the fractional cover number of $Z \cap C$ is at most half the fractional cover number of $Z$.  The idea now is to find a small and balanced separator $S$ for $Z$, for each connected component $C$ of $G - S$ solve the problem recursively on $H[N[C]]$ with new root $Z_C = (Z \cap C) \cup S$, and stitch together the tree decompositions of the recursive calls into a tree decomposition of $H$.

The main issue faced by our algorithm (and every algorithm based on this scheme) is that only {\em existence} of the small and balanced separator is guaranteed from the assumption that $\fhtw(H)$ is small. We deal with this in the usual way - if we make the condition that $Z$ is ''large enough'' stronger (we require that the fractional cover number of $Z$ is at least $\Omega(\omega \cdot \log \omega \cdot \log n)$) then we can use an approximation algorithm for finding a ``small enough'' and ``balanced enough'' separator $S'$ which can take the place of $S$ in the recursive scheme above. The approximation ratio of the overall algorithm (in this case $\cO(\log \omega \cdot \log n)$) ends up being equal to (up to constant factors) the approximation ratio of the algorithm to find $S'$.

Thus we are faced with the task of designing an algorithm for the following problem. Input is a hypergraph $H$, and a vertex set $Z$ in $H$. The task is to find a set $S$ with minimum fractional cover number, such that for every connected component $C$ of $H - S$ it holds that the fractional cover number of $Z \cap C$ is at most half the fractional cover number of $Z$.

Our approach to this balanced separator problem is to formulate a Linear Program (LP). This is the place where our algorithm takes a very different path than the approach of Marx~\cite{DBLP:journals/talg/Marx10}. Our linear programming formulation very closely follows the classic balanced separator LP formulated by Leighton and Rao~\cite{leighton1999multicommodity}. The main twist is in how we measure that the separator is ``small''. Recall that ``small'' means that the fractional cover of $S$ must be small. We encode this in the linear program in the most natural way - by explicitly having variables in the LP that select the fractional cover $\gamma : E(H) \rightarrow [0,1]$ of $S$, while minimizing $\sum_{e \in E(H)} \gamma(e)$.

The algorithm follows the (by now) standard scheme~\cite{RobertsonS95b} for computing tree decompositions, and the main twist is to use a (by now) standard linear program~\cite{leighton1999multicommodity} for balanced separators. 
However, most of the time, when one replaces a ``direct'' cost function (measuring the size of the separator as the number of vertices) by an indirect one (such as measuring cost in terms of the fractional cover), the integrality gap of the LP blows up so much that it becomes unusable for reasonable approximation algorithms. So it was quite a surprise to us that for this LP it is possible to achieve a rounding scheme with gap $\cO(\log \omega \cdot \log n)$, and in fact even smaller gap for some interesting special cases. 

Our rounding scheme first uses a ball growing argument similar to that of Leighton and Rao~\cite{leighton1999multicommodity}. This step incurs a $\cO(\log \omega)$ gap, and leaves us with the following setting. We have a hypergraph $H$, and two vertex sets $A$ and $B$ in $H$, together with a fractional $A$-$B$ separator $x$. We need an $A$-$B$ separator $X$ such that the fractional cover number of $X$ (with respect to $E(H)$) is not much larger than the fractional cover number of $x$. This is precisely the setting of Lemma~\ref{lem:mainCliqueSepSimplified}, which guarantees that one can find an $A$-$B$ separator $X$ such that the fractional cover number of $X$ is at most a factor $\cO(\log n)$ larger than the fractional cover number of $x$. We now give a sketch of the proof of Lemma~\ref{lem:mainCliqueSepSimplified}.

\begin{proof}[Proof sketch of Lemma~\ref{lem:mainCliqueSepSimplified}]
Given $H$ and $x$, observe first that without loss of generality there is no vertex $v$ such that $0 < x(v) \leq \frac{1}{n}$. In particular, for every $v$ such that $x(v) \leq \frac{1}{2n}$ we can change $x(v)$ to $0$. For all other vertices we can double $x(v)$. It is easy to see that this new $x$ is still a fractional $A$-$B$ separator, and that this operation at most doubles $\rho_{E(H)}^*(x)$. 

Now, let $\gamma : E(H) \rightarrow [0, 1]$ be such that
\begin{equation}
    \sum_{\substack{f \in E(H) \mbox{ s.t. } \\ v \in f}} \gamma(f) \geq x(v)
    \label{eqn: gamma_cover}
\end{equation}
and $\rho_{E(H)}^*(x) = \sum_{e \in E(H)} \gamma(e)$. Define for every vertex a real $d_v$, which is the minimum $\sum_{u \in V(P)} x(u)$ taken over all paths $P$ from $A$ to $v$.
Since $x$ is a fractional $A$-$B$ separator we have that $d_v \geq 1$ for every $v \in B$. Furthermore, the $d_v$'s satisfy a kind of triangle inequality: for every pair of adjacent vertices $u$ and $v$ we have that $d_v \leq d_u + x(v)$. Sample $r \in [0, 1]$ uniformly at random and output 
$$X_r = \{v \in V(H) : d_v - x(v) \leq r \leq d_v\}$$

First, observe that $X_r$ is an $A$-$B$ separator for every $r \in [0, 1]$. Indeed, for every $A$-$B$ path $P$ the first vertex $v$ such that $d_v \geq r$ must be in $X_r$. This follows from the triangle inequality applied to $v$ and the predecessor of $v$ on $P$.
%
%
Next we show that $$E[\rho_{E(H)}^*(X_r)] \leq  \rho_{E(H)}^*(x) \cdot \cO(\log n)$$
To this end we  construct (given $r$) a function $\gamma_r : E(H) \rightarrow [0, 1]$. For every $e \in E(H)$ such that $e \cap X_r = \emptyset$ we set $\gamma_r(e) = 0$. Otherwise we set 
\begin{equation}
    \gamma_r(e) = \frac{\gamma(e)}{\min_{v \in X_r \cap e} x(v)}
    \label{eqn: weighted_cover}
\end{equation}
Fix a vertex $v$. It follows from (\ref{eqn: weighted_cover}) that  $\gamma_r(f) \geq \frac{\gamma(f)}{x(v)}$ for every $f$ containing $v$. Therefore, by (\ref{eqn: gamma_cover}) we have 
$$\sum_{\substack{f \in E(H) \mbox{ s.t. } \\ v \in f}} \gamma_r(f) \geq \frac{x(v)}{x(v)} = 1$$
Hence, to upper bound $E[\rho_{E(H)}^*(X_r)]$ it suffices to upper bound $E[\sum_{e \in E(H)} \gamma_r(e)]$, which by linearity of expectation is equal to $\sum_{e \in E(G)} E[\gamma_r(e)]$.

The crucial observation is that the probability that $\gamma_r(e)$ is at least $\gamma(e) \cdot 2^i$ is at most $2^{1-i}$.
Indeed, if $\gamma_r(e) \geq \gamma(e) \cdot 2^i$ then $X_r \cap e$ must contain some vertex $v$ such that $x(v) \leq 2^{-i}$. 
Among all the vertices in $e$ with $x(v) \leq 2^{-i}$, choose $u$ and $v$ such that $d_u - x(u)$ is minimized and such that $d_v$ is maximized. 
Then $\gamma_r(e) \geq \gamma(e) \cdot 2^i$ implies that $d_u - x(u) \leq r \leq d_v$. But $u$ and $v$ are adjacent (since they are both in $e$) and hence $d_v \leq d_u + x(v) \leq d_u - x(u) + 2 \cdot2^{-i}$. It follows that the probability that $\gamma_r(e)$ is at least $\gamma(e) \cdot 2^i$ is at most $2 \cdot 2^{-i}$. We can now conclude the analysis as follows.
\begin{align*}
\sum_{e \in E(G)} E[\gamma_r(e)] & \leq \sum_{e \in E(G)} \gamma(e) \cdot \sum_{i=0}^{\log n} 2^{i+1} \cdot \Pr[\gamma(e) \cdot 2^{i+1} \geq \gamma_r(e) \geq \gamma(e) \cdot 2^i] \\
& \leq  \sum_{e \in E(G)} \gamma(e) \cdot \sum_{i=0}^{\log n} 2^{i+1} \cdot 2^{1-i} \\
& \leq  \sum_{e \in E(G)} \gamma(e) \cdot 4 \log n \\
& \leq \rho_{E(H)}^*(x) \cdot 4 \log n
\end{align*}
Here $i$ ranges from $0$ to $\log n$ because every non-zero $x(v)$ is at least $1/n$, from which it follows that $\gamma_r(e) \leq \gamma(e) \cdot n$.
\end{proof}

The $\cO(\log n)$ factor incurred in Lemma~\ref{lem:mainCliqueSepSimplified} together with the  $\cO(\log \omega)$ factor incurred by the ball growing argument leads to the $\cO(\log \omega \log n)$-approximation ratio of our algorithm for finding balanced separators, and therefore also for the algorithm of Theorem~\ref{theorem:main-poly-approx}. The Clique Menger's Theorem (Theorem~\ref{thm:mainMengerSimplified}) follows quite directly from Lemma~\ref{lem:mainCliqueSepSimplified} together with standard path sampling arguments~\cite{DBLP:journals/combinatorica/RaghavanT87}.

The stronger upper bounds in the statement of Theorem~\ref{theorem:main-poly-approx} come from a sligthly tighter analysis of the rounding scheme of Lemma~\ref{lem:mainCliqueSepSimplified} when the input hypergraph is assumed to have some additional properties. 
For an example when the maximum independent set size $\alpha(H)$ in $H$ is small we can assume without loss of generality that every non-zero $x(v)$ is at least $\frac{1}{2\alpha(H)}$ by an argument identical to the one in the beginning of the proof sketch of Lemma~\ref{lem:mainCliqueSepSimplified}. Then the sum at the end of the proof sketch only needs to range from $i=0$ to $i=\log (2\alpha(H))$, leading to a rounding scheme which incurs a factor of $\cO(\log(\alpha(H))$ in the cost instead of $\cO(\log n)$.
When the degeneracy $\mu_H$ of the bipartite incidence graph of $H$ is bounded we can achieve a rounding scheme that incurs a factor of $\cO(\mu_H)$ by carefully re-defining $\gamma_r(e)$ using the degeneracy sequence of the bipartite incidence graph of $H$.
For the last bound in Theorem~\ref{theorem:main-poly-approx} we show that for every hypergraph $H$ of fractional hypertree width at most $\omega$ such that every pair of hyperedges in $H$ intersect in at most $\eta_H$ vertices, the degeneracy of the bipartite incidence graph of $H$ is upper bounded by $\cO(\omega \cdot \eta_H)$ (see Lemma~\ref{lemma:intersection_bound}). 

\medskip
We now briefly discuss the algorithm of Theorem~\ref{theorem:main-FPT-approx}. The algorithm proceeds in multiple iterations. In each iteration the algorithm has a tree decomposition $(\hat{T}, \hat{\beta})$ of $H$ with fractional hypertree width at most $\hat{\omega}$, and the goal is to produce a tree decomposition whose fractional hypertree width is at most $\hat{\omega} - 1$. Initially this tree decomposition is the trivial one with a single bag containing all the vertices.

Each iteration of the algorithm of Theorem~\ref{theorem:main-FPT-approx} follows the exact same template as the algorithm of Theorem~\ref{theorem:main-poly-approx}. The only difference is in the approximation algorithm for balanced separators.
We know that there exists a balanced separator $S$ for $Z$ such that the fractional cover number of $S$ is at most $\omega$. It is not too hard to show that $S$ must then be contained in the union of at most $\lfloor \omega \rfloor$ bags of $(\hat{T}, \hat{\beta})$. Since we have $(\hat{T}, \hat{\beta})$ available we can iterate over all $n^{\lfloor \omega \rfloor}$ choices of $\lfloor \omega \rfloor$ bags of $(\hat{T}, \hat{\beta})$ bags. In each iteration we set $R$ to be the union of the selected bags, and search for a balanced separator for $Z$ inside $R$. We know that in at least one iteration the balanced separator $S$ is a subset of $R$.
Each bag of $(\hat{T}, \hat{\beta})$ has fractional cover $\hat{\omega}$, and therefore independence number at most $\hat{\omega}$ as well (see Observation~\ref{obs: mis_cov}). It follows that $\alpha(H[R]) \leq \omega \cdot \hat{\omega}$. When we solve the LP for the balanced separator problem we insist that only vertices in $R$ can be (fractionally) picked. Now, when we round the balanced separator LP and invoke Lemma~\ref{lem:mainCliqueSepSimplified} the factor incurred by this rounding is  $\cO(\log(\alpha(H[R]))) \leq \cO(\log(\omega \hat{\omega}))$ instead of $\cO(\log n)$. Thus the iteration outputs a tree decomposition with fractional hypertree width upper bounded by $\cO(\omega \log(\omega)\log(\hat{\omega}))$. If the new tree decomposition is better than the old one, we proceed to the next iteration. Otherwise, we have that $\hat{\omega} \leq c \cdot \log(\omega) \log (\hat{\omega})$ for some fixed constant $c$, which in turn implies $\hat{\omega} = \cO(\omega \log^2 \omega)$ and we end our algorithm.

\section{Notation and Preliminaries}\label{sec:prelims}
\subsection{Definitions}
\noindent \textbf{Hypergraphs.} A {\em hypergraph} is a pair $H=(V(H),E(H))$, consisting of a set $V(H)$ of vertices and a set $E(H)$ of subsets of $V(H)$, the hyperedges of $H$. We always assume that hypergraphs have no isolated vertices, that is, for every $v\in V(H)$ there exists at least one $e\in E(H)$ such that $v\in e$. For each $v \in V(H)$, we write $E_v(H) = \{e \in E(H): v \in e\}$ and for each $Z\subseteq V(H)$, we write $E_Z(H) = \{e \in E(H): e\cap Z\neq \emptyset\}$.
Let $\norm{H}:=|V(H)|+|E(H)|$, we will express the running time of the algorithms as a function of $\norm{H}$.

For a hypergraph $H$ and a set $X\subseteq V(H)$, the subhypergraph of $H$ induced by $X$ is the hypergraph $H[X] = (X,\{e\cap X|e\in E(H)\})$.
We let $H\setminus X=H[V(H)\setminus X]$ and $H\setminus e$ be the graph obtained by removing $e$ from $H$. The primal or Gaifman graph of a hypergraph $H$ is the graph
$\gf{H} = (V(H),\{\{v,u\} \text{ $|$ } v \neq u, \text{ there exists an } e \in E(H) \text{ such that } \{v,u\}\subseteq e\})$.

A hypergraph $H$ is connected if $\gf{H}$ is connected. A set $C \subseteq V(H)$ is connected (in $H$) if the induced subhypergraph $H[C]$ is connected, and a connected component of $H$ is a maximal connected subset of $V(H)$. A sequence of vertices of $H$ is a path of $H$ if it is a path of $\gf{H}$. A subset $K\subseteq V(H)$ is a clique of $H$ if $K$ induces a clique in $\gf{H}$. 

A graph $G$ is $d$-\emph{degenerate} if for every induced subgraph, there exists a vertex with degree at most $d$.
The \emph{degeneracy} of a graph $G$ is a minimum number $d$ such that $G$ is $d$-degenerate.

A \emph{degeneracy ordering} is some ordering on the vertices of $G$ that can be obtained by repeatedly removing a vertex of minimum degree in the remaining subgraph. In particular, it follows from the definition of degeneracy that every vertex has at most $d$ neighbors that appear later than it in the degeneracy ordering.
The \emph{degeneracy} of a hypergraph $H$ (which we denote by $\mu_H$) is the degeneracy of its vertex-hyperedge incidence graph $I_H$.

We denote by $\eta_H$ the maximum number of vertices in the intersection of any two hyperedges in $H$, i.e. $\eta_H=\max_{e_1,e_2\in E(H):e_1\neq e_2}|e_1\cap e_2|$.

\medskip \noindent \textbf{Edge Covers.}
Let $H=(V(H),E(H))$ be a hypergraph and let $S$ be a subset of its vertices. 

\begin{definition}
An \emph{edge cover} of $S$ is a set $F \subseteq E(H)$ such that for every $v\in S$, there is an $e\in F$ with $v \in e$.
The size of the smallest edge cover of $S$, denoted by $\rho_H(S)$, is the \emph{edge cover number} of $S$. A \emph{fractional edge cover} of $S \subseteq V (H)$ is a mapping $\gamma : E(H) \rightarrow [0,1]$ such that for every $v \in S$, we have
$\sum_{e\in E(H): v\in e}\gamma(e) \geq 1$. The \emph{cost} of the assignment 
$\gamma$ is $\cost(\gamma) := \sum_{e\in E(H)}\gamma(e)$.The \emph{fractional edge cover number} of $S$, denoted by $\cov(S)$, is the minimum of $\cost(\gamma)$ taken over all fractional edge covers of $S$.  
\end{definition}
It is well known that $\cov(S)\leq \rho_H(S)\leq  \cov (S)(1+\ln|V(H)|)$; in fact, a simple greedy algorithm can be used to find an edge cover of $S$ with size at most $\cov(S)(1+\ln|V(H)|)$(~\cite{DBLP:journals/talg/Marx10}).  Note that determining $\rho_H(S)$ is NP-hard, while $\cov(S)$ can be determined in polynomial time using linear programming. We define $\rho(H)$ and $\rho^*(H)$ to be  $\rho_H(V (H ))$ and  $\cov(V (H ))$, respectively. 

The above definitions speak about (fractional) edge covers of a given subset $S$ of $V(H)$. Similarly, we define a (fractional) edge cover of any function $x: V(H) \to [0,1]$. The \emph{support} of $x$ is defined as $\supp(x) = \{v \in V(H): x(v) > 0\}$.
\begin{definition} \label{def-fracec}
Let $x:V(H)\rightarrow [0,1]$ be a function.
A \textbf{fractional edge cover} of $x$ is a function $y: E(H) \rightarrow [0,1]$ satisfying that $\sum_{e \in E(H): v \in e} y(e) \geq x(v)$ for all $v \in V(H)$.
The \textbf{cost} of $y$ is $\sum_{e \in E(H)} y(e)$.
The \textbf{fractional edge cover number} of $x$, denoted by $\rho_H^*(x)$, is the minimum cost of a fractional edge cover of $x$.
\end{definition}

In particular, the case of a given subset $S$ corresponds to choosing $x$ as a characteristic function of $S$, that is, $\rho_H^*(S) = \rho_H^*(\mathbf{1}_S)$, where the function $\mathbf{1}_S:V(H)\rightarrow [0,1]$ is defined as $\mathbf{1}_S(v) = 1$ for $v \in S$ and $\mathbf{1}_S(v) = 0$ for $v \in V(H) \backslash S$.

\smallskip \noindent \textbf{$(A,B)$-separators.}
Let $A$ and $B$ be subsets of vertices of $H$. Similarly to the case of edge covers, we will consider standard and fractional $(A,B)$-separators. 

\begin{definition}
An \textbf{$(A,B)$-separator} in $H$ is a set $S \subseteq V(H)$ such that there is no path connecting a vertex of $A\setminus S$ with a vertex of $B\setminus S$ in the hypergraph $H \setminus S$. 
\end{definition}

For a hypergraph $H$ and two non-adjacent vertices $a,b\in V(H)$, refer to $(\{a\},\{b\})$-separators as $(a,b)$-separators. We denote by $\optab$ the fractional edge cover $\cov(S)$ of an $(a,b)$-separator with the least fractional edge cover among all $(a,b)$-separators. 

To define fractional $(A,B)$-separator, we will need the notion of a distance function.

\begin{definition}
    Let $x:V(H)\rightarrow [0,1]$ be a function. The \textbf{distance function} $\mathsf{dist}_{H,x}:V(H)\times V(H)\rightarrow {\mathbb{R}}_{>0}$ and $\mathsf{dist}^*_{H,x}:V(H)\times V(H)\rightarrow [0,1] $ are defined as: for each $u,v\in V(H)$, $\mathsf{dist}_{H,x}(u,v)=\min_{\text{$u$-$v$ path $P$ in $\gf{H}$}}\sum_{w\in V(P)}x(w)$ and $\mathsf{dist}^*_{H,x}(u,v)=\min\{\mathsf{dist}_{H,x}(u,v),1\}$. 
    
A \textbf{fractional $(A,B)$-separator} in $H$ is a function $x: V(H) \rightarrow [0,1]$ satisfying $\mathsf{dist}^*_{H,x}(u,v) \geq 1$ for all $(u,v) \in A \times B$.
\end{definition}
Note that if a subset $S$ of vertices of $H$ is an $(A,B)$-separator, then the characteristic function of $S$ is a fractional $(A,B)$-separator.  

\medskip \noindent \textbf{Tree Decompositions and Width Measures.} A {\em tree decomposition} of a hypergraph $H$ is a pair $(T, \beta)$ where $T$ is a tree and $\beta : V(T) \rightarrow 2^{V(H)}$ is a function that assigns to each vertex $u$ of $T$ a set $\beta(u)$ (called a {\em bag}) of vertices of $H$. To be a tree decomposition of $H$ the pair $(T, \beta)$ must satisfy the following two axioms:  {\em (i)} For every $v \in V(H)$ the set $\{u \in V(T) : v \in \beta(u)\}$ induces a non-empty and connected sub-tree of $T$, and {\em (ii)} For every $e \in E(H)$ there exists a $u \in V(T)$ such that $e \subseteq \beta(u)$.
The width of a tree decomposition $(T, \beta)$ is $\max \{|\beta(t)| : t \in V (T)\} - 1$. The \emph{treewidth} $\tw(H)$ of a hypergraph $H$ is the minimum of the widths of all tree decompositions of $H$. It is easy to see that $\tw(H) = \tw(\gf{H})$.

The generalized hypertree width of a decomposition $(T, \beta)$ is $\max \{\rho_H(\beta(t)) : t \in V (T)\}$ and the \emph{generalized hypertree width} of a hypergraph $H$, denoted by $\ghtw(H)$, is the minimum of the generalized hypertree widths of all tree decompositions of $H$. \emph{Fractional hypertree width} of a tree decomposition and of a hypergraph is defined analogously, by having $\cov(\beta(t))$ instead of $\rho_H(\beta(t))$ in the definition. We denote by $\fhtw(H)$ or $\omega(H)$ the fractional hypertree width of $H$. When the graph is clear from context, we will simply use $\omega$.

\subsection{Fractional Edge Covers and Maximum Independent Sets} Bounded fractional hypertree width imposes restrictions on the size and structure of independent sets in $H$. A vertex set $I \subseteq V(H)$ is {\em independent} if no hyperedge of $H$ contains at least two vertices in $I$. For a subset $S$ of vertices of $H$, let $\mis(S)=\mis_H(S)$ denote the size of the maximum independent set of $S$. Note that for any independent set $F\subseteq S$, the hyperedge subsets covering the vertices of $F$ are pairwise disjoint. The sums of edge weights in every such subset should be at least one to cover the corresponding vertex of $F$. In particular, $\cov(S)\geq |F|$:

\begin{observation}
\label{obs: mis_cov}
 Let $S$ be a subset of vertices of $H$, then $\cov(S)\geq \mis(S)$.
\end{observation}

 \begin{lemma}
 If $H$ has fractional hypertree width $\omega$ and $S$ is a subset of vertices of $H$, then $\cov(S) \leq \mis(S) \cdot \omega$.  
 \end{lemma}
 \begin{proof}
  Let $\mathcal T=(T, \beta)$ be the tree decomposition of $H$ of fractional hypertree width $\omega$. Let $i_1$ be the vertex of $S$ whose highest appearance in $\mathcal T$ is the lowest among all the vertices of $S$. Let $t$ be the highest node of $T$ such that $i_1\in \beta(t)$. Then $\beta(t)$ contains all the neighbors of $i_1$ in $S$. Indeed, as $i_1$ only occurs in the subtree rotted at $t$, it has no neighbors that are introduced outside the subtree rooted in $t$. On the other hand, if some neighbor $i'\in S$ of $i_1$ would appear only below $t$, it would mean that the highest appearance of $i'$ in $\mathcal T$ is lower than the one of $i_1$, contradicting to our choice of $i_1$. Therefore, all the neighbors of $i_1$ in $S$ belong to $\beta(t)$ and hence admit a fractional edge cover $\mathcal E^1$ with $\cost(\mathcal E^1) \leq \omega$. We delete the neighborhood of $i_1$ in $S$ (i.e., the set $N_S(i_1)=\{i_1\} \cup \{v\in V(H)~|~\exists e\in E(H) \text{ s.t. } \{i_1, v\} \subseteq e\})$ from all the bags of $\mathcal T$ and repeat similar procedure by choosing the vertex $i_2 \in S$ whose highest appearance is the lowest in the resulting tree decomposition. Let $\mathcal E^2$ be some fractional edge cover of the neighborhood of $i_2$ in $S \setminus N_S(i_1)$ such that $\cost(\mathcal E^2) \leq \omega$. We keep choosing $i_j$ and $\mathcal E^j$, $j \geq 1$, this way until there are no vertices of $S$ left in any bag of a tree decomposition. By construction, the set of all $i_j$ form an independent set in $S$, so we will do $l \leq \mis(S)$ iterations. Let $\mathcal E:E(H) \to [0,1]$ be defined by $\mathcal E(e)=\max_{j\in [l]} \mathcal E^j(e)$, $e\in E(H)$. Then $\mathcal E$ is a fractional edge cover of $S$. Since $\mathcal E(e)=\max_{j\in [l]} \mathcal E^j(e) \leq \sum_{j\in [l]} \mathcal E^j(e)$, we have $\cost(\mathcal E) = \sum_{e\in E(H)} \mathcal E(e) \leq \sum_{j\in [l]} \cost(\mathcal E^j) \leq l\omega \leq \mis(S) \cdot \omega$.  
 \end{proof}

The construction of a fractional edge cover from the last proof can be generalised to any tree decomposition $\mathcal T$, not necessarily of the optimal width. In any case it allows to cover $S$ by at most $\mis(S)$ bags.
Since $\mis(S)\leq \cov(S)$ by Observation \ref{obs: mis_cov}, we immediately obtain:
\begin{corollary}
\label{corollary:cov_bags}
      Given a hypergraph $H$ along with its tree decomposition $\mathcal T$, any subset $S$ of vertices of $H$ can be covered by $\cov (S)$ bags of $\mathcal T$. 
\end{corollary}

\section{Computing $(A,B)$-separators}
\label{sect:ab-sep}

In this section, we consider $(A,B)$-separators and their fractional counterparts in a hypergraph $H$.
We study how to compute an $(A,B)$-separator with small fractional edge cover number.
The main result, presented in Theorem~\ref{theorem:ab-sep-main}, is an algorithm which rounds a given fractional $(A,B)$-separator $x: V(H) \rightarrow [0,1]$ to an $(A,B)$-separator $S$ such that $\rho_H^*(S)$ is not much larger than $\rho_H^*(x)$.
This rounding algorithm will serve an important role in the algorithm for computing balanced separators (Theorem~\ref{theorem:bal-separator-main}), which is in turn the key ingredient of our main algorithms for computing fractional hypertree width.
The algorithm also directly results in an approximation algorithm for computing an $(A,B)$-separator with minimum fractional edge cover number, which we discuss in Section~\ref{sec-computingAB}.
Finally, using (a stronger version of) the rounding theorem, we can prove our clique version of Menger's Theorem (Section~\ref{sec-Mengers}).

\subsection{Rounding algorithm} \label{sec-ABrounding}

In this section, we present our rounding algorithm for $(A,B)$-separators in a hypergraph $H$.
Formally, our result is presented in the following theorem.

\begin{theorem}\label{theorem:ab-sep-main}
Given a hypergraph $H$, two sets $A,B \subseteq V(H)$, and a fractional $(A,B)$-separator $x: V(H) \rightarrow [0,1]$ with $\supp(x) = R$, one can compute in $\lVert H \rVert^{O(1)}$ time an $(A,B)$-separator $S \subseteq R$ satisfying that $\rho_H^*(S) \leq \min\{8 + 4 \ln \alpha_H(R), 6 \mu_H\} \cdot \rho_H^*(x)$.
\end{theorem}

Although the above theorem is already sufficient for the purpose of computing fractional hypertree width, we actually prove a slightly stronger result, which will be used later for proving our clique version of Menger's Theorem.
To state this stronger result, we need to slightly extend the notion of fractional edge covers.
For a subset $E \subseteq E(H)$ of hyperedges, a \textit{fractional $E$-edge cover} of a function $x:V(H) \rightarrow [0,1]$ is a function $y:E \rightarrow [0,1]$ such that $\sum_{e \in E: v \in e} y_e \geq x(v)$ for all $v \in V(H)$.
The \textit{cost} of $y$ is $
\cost(y)=\sum_{e \in E} y(e)$.
The \textit{fractional $E$-edge cover number} of $x:V(H) \rightarrow [0,1]$, denoted by $\rho_{H,E}^*(x)$, is the minimum cost of a fractional $E$-edge cover of $x$; when there does not exist a fractional $E$-edge cover, we simply define $\rho_{H,E}^*(x) = \infty$.
For a set $S \subseteq V(H)$, we can also define fractional $E$-edge covers of $S$ as well as the fractional $E$-edge cover number $\rho_{H,E}^*(S)$ by considering the function $\mathbf{1}_S:V(H)\rightarrow [0,1]$ defined as $\mathbf{1}_S(v) = 1$ for $v \in S$ and $\mathbf{1}_S(v) = 0$ for $v \in V(H) \backslash S$, as what we did in Definition~\ref{def-fracec}.
Clearly, when $E = E(H)$, the notion of fractional $E$-edge covers coincides with the notion of fractional edge covers.
The following theorem implies Theorem~\ref{theorem:ab-sep-main}.

\begin{theorem}\label{theorem:ab-sep-stronger}
Given a hypergraph $H$, two sets $A,B \subseteq V(H)$, a fractional $(A,B)$-separator $x: V(H) \rightarrow [0,1]$ with $\supp(x) = R$, and a set $E \subseteq E(H)$, one can compute in $\lVert H \rVert^{O(1)}$ time an $(A,B)$-separator $S \subseteq R$ satisfying that $\rho_{H,E}^*(S) \leq \min\{8 + 4 \ln \alpha_H(R), 6 \mu_H\} \cdot \rho_{H,E}^*(x)$.
\end{theorem}

In the remainder of the section, we present the proof of Theorem~\ref{theorem:ab-sep-stronger}.

\paragraph{The main rounding algorithm.}
Given $H,A,B,x,E$ as in Theorem~\ref{theorem:ab-sep-stronger}, we show how to compute the desired separator $S$. Let $y: E \rightarrow [0,1]$ be a fractional $E$-edge cover of $x$ of the minimum cost, that is, $\cost(y)=\rho_{H,E}^*(x)$. Such a fractional $E$-edge cover can be computed via simple LP, so we can assume that $y$ is provided as well.

For convenience, we first construct a new hypergraph $H'$ from $H$ by adding two dummy vertices $a,b$ together with edges $\{a,a'\}$ for $a' \in A$ and $\{b,b'\}$ for $b' \in B$.
We extend the domain of the function $x$ from $V(H)$ to $V(H')$ by simply setting $x(a) = x(b) = 0$.
Observe that $\mathsf{dist}^*_{H',x}(a,b) \geq 1$.
For each vertex $v \in V(H')$, we write $d_v = \mathsf{dist}^*_{H',x}(a,v)$ and define an interval $I_v = [d_v - x(v), d_v]$.
For a number $r \in \mathbb{R}$, define $S_r = \{v \in V(H'): r \in I_v \text{ and } x(v) > 0\}$.

\begin{observation} \label{obs-stab}
    For every $e \in E(H')$, $\bigcap_{v \in e} I_v \neq \emptyset$.
\end{observation}
\begin{proof}
Note that for any $u,v \in e$, we have the inequality $d_u \leq d_v + x(u)$, which implies $d_v \in I_u$ and $I_u \cap I_v \neq \emptyset$.
By Helly's theorem, pairwise intersecting intervals have a nonempty common intersection.
Thus, we have $\bigcap_{v \in e} I_v \neq \emptyset$.
\end{proof}

\begin{observation}
    For every $r \in [0,1]$, $S_r \subseteq R$ and $S_r$ is an $(A,B)$-separator in $H$.
\end{observation}
\begin{proof}
We have $S_r \subseteq R$ simply because $x(v) > 0$ for all $v \in S_r$.
To see $S_r$ is an $(A,B)$-separator, consider a path $(v_1,\dots,v_{\ell-1})$ in $H$ where $v_1 \in A$ and $v_{\ell-1} \in B$.
We need to show that $v_i \in S_r$ for some $i \in [\ell-1]$.
Set $v_0 = a$ and $v_\ell = b$.
Then $(v_0,v_1,\dots,v_\ell)$ is a path in $H'$ connecting $a$ and $b$.
For $i \in [\ell]$, since $\{v_{i-1},v_i\} \in e$ for some edge $e \in E(H)$, we have $d_{v_i} \leq d_{v_{i-1}} + x(v_i)$, i.e., $d_{v_i} - x(v_i) \leq d_{v_{i-1}}$.
Note that $a,b \notin S_r$ because $x(a) = x(b) = 0$.
Therefore, it suffices to show the existence of $i \in \{0\} \cup [\ell]$ satisfying $v_i \in S_r$.
We distinguish two cases: $r > 0$ and $r = 0$.
Suppose $r > 0$.
Let $i \in \{0\} \cup [\ell]$ be the smallest index such that $d_{v_i} \geq r$.
Such an index must exist, as $d_{v_\ell} = d_b \geq 1 \geq r$.
Since $d_{v_0} = d_a = 0$, we have $i \geq 1$.
Furthermore, $d_{v_{i-1}} < r$ by the choice of $i$.
We claim that $v_i \in S_r$.
As observed above, $d_{v_i} - x(v_i) \leq d_{v_{i-1}} < r$.
Because $d_{v_i} \geq r$, we have $x(v_i) > 0$ and $r \in [d_{v_i} - x(v_i),d_{v_i}] = I_{v_i}$, which implies $v_i \in S_r$.
Next, suppose $r = 0$.
Let $i \in \{0\} \cup [\ell]$ be the smallest index such that $d_{v_i} > 0$.
Again, such an index exists since $d_{v_\ell} \geq 1 > 0$, and $i \geq 1$ since $d_{v_0} = 0$.
We have $d_{v_{i-1}} = 0$, which implies $d_{v_i} - x(v_i) = 0$ because $d_{v_i} - x(v_i) \leq d_{v_{i-1}}$.
It follows that $x(v_i) > 0$ and $0 \in [d_{v_i} - x(v_i),d_{v_i}] = I_{v_i}$, which implies $v_i \in S_0 = S_r$.
\end{proof}

Our rounding algorithm simply picks a number $r^* \in [0,1]$ that minimizes $\rho_{H,E}^*(S_{r^*})$, and then return $S = S_{r^*}$, which is an $(A,B)$-separator by the above observation.
Note that although $[0,1]$ is an infinite set, such an $r^*$ always exists and can be found in $\lVert H \rVert^{O(1)}$ time.
Indeed, when $r$ moves on the real line $\mathbb{R}$, the set $S_r$ only changes when $r$ reaches the endpoints of the intervals $I_v$.

In what follows, we shall give two different analysis of rounding our algorithm - one based on the minimum nonzero value of $x$ and another based on the degeneracy of the incidence graph $H$. We do this to get the different approximation bound guarantees. Finally we combine the two parts of our analysis to prove Theorem~\ref{theorem:ab-sep-stronger}.

\paragraph{Analysis based on the minimum nonzero value of $x$.}
Let $\varepsilon \in (0,1]$ be the largest number such that $x(v) \in \{0\} \cup [\varepsilon,1]$ for all $v \in V(H)$.
In other words, $\varepsilon$ is the minimum nonzero value of the function $x$.
In the first part of our analysis, we shall show that $\rho_{H,E}^*(S_{r^*}) \leq (1+2\ln \frac{1}{\varepsilon}) \cdot \rho_{H,E}^*(x)$.
Given $E$ and $y$, for every $r \in [0,1]$, we define a canonical fractional $E$-edge cover $y_r: E \rightarrow [0,1]$ of $S_r$ as follows.
Let $m_r(e) = \max\{\frac{1}{x(v)}: v \in S_r \cap e\}$ for $e \in E$; if $S_r \cap e = \emptyset$, simply set $m_e = 0$.
Then we define $y_r(e) = \min\{m_r(e) \cdot y(e),1\}$ for all $e \in E$.
\begin{observation}
    $y_r$ is a fractional $E$-edge cover of $S_r$.
\end{observation}
\begin{proof}
Consider a vertex $v \in S_r$.
We have $x(v) > 0$ by the definition of $S_r$.
We need to show $\sum_{e \in E: v \in e} y_r(e) \geq 1$.
If there exists an hyperedge $e \in E$ with $v \in e$ and $y_r(e) = 1$, we are done.
Otherwise, for each $e \in E$ with $v \in e$, we have $y_r(e) = m_r(e) \cdot y(e) \geq y(e)/x(v)$, since $v \in S_r \cap e$.
Thus, $\sum_{e \in E: v \in e} y_r(e) \geq \sum_{e \in E: v \in e} y(e)/x(v)$.
As $y$ is a fractional $E$-edge cover of $x$, we have $\sum_{e \in E: v \in e} y(e) \geq x(v)$, i.e., $\sum_{e \in E: v \in e} y(e)/x(v) \geq 1$, which implies $\sum_{e \in E: v \in e} y_r(e) \geq 1$.
\end{proof}

Let $r \in [0,1]$ be a random number sampled from the uniform distribution on $[0,1]$.
By the definition of $r^*$ and the above observation, we have
\begin{equation*}
    \rho_{H,E}^*(S_{r^*}) \leq \mathbb{E}[\rho_{H,E}^*(S_r)] \leq \mathbb{E}\left[\sum_{e \in E} y_r(e)\right] \leq \mathbb{E}\left[\sum_{e \in E} m_r(e) \cdot y(e)\right] = \sum_{e \in E} \mathbb{E}[m_r(e)] \cdot y(e).
\end{equation*}
Next, we bound $\mathbb{E}[m_r(e)]$ for each hyperedge $e \in E$.
By Observation~\ref{obs-stab}, $\bigcap_{v \in e} I_v \neq \emptyset$ and we can pick a number $q \in \bigcap_{v \in e} I_v$.
\begin{observation}
    If $m_r(e) \geq t$ for a number $t > 0$, then $r \in [q-\frac{1}{t},q+\frac{1}{t}]$.
\end{observation}
\begin{proof}
Note that for every $v \in e$ with $x(v) \leq \frac{1}{t}$, we have $I_v \subseteq [q-\frac{1}{t},q+\frac{1}{t}]$, simply because $q \in I_v$ and the length of $I_v$ is $x(v)$, which is at most $\frac{1}{t}$.
If $m_r(e) \geq t$, then there exists $v \in S_r \cap e$ with $x(v) \leq \frac{1}{t}$.
Since $v \in S_r$, we have $r \in I_v \subseteq [q-\frac{1}{t},q+\frac{1}{t}]$.
\end{proof}

\noindent
For $t > 0$, the probability of the event $r \in [q-\frac{1}{t},q+\frac{1}{t}]$ is at most $\frac{2}{t}$.
Therefore, by the above observation, $\Pr[m_r(e) \geq t] \leq \frac{2}{t}$.
Since $m_r(e) \leq \frac{1}{\varepsilon}$ for all $r \in [0,1]$, we have
\begin{align*}
    \mathbb{E}[m_r(e)] & = \int_0^{\frac{1}{\varepsilon}} \Pr[m_r(e) \geq t]\ \text{d}t \\
    & = \int_0^1 \Pr[m_r(e) \geq t]\ \text{d}t + \int_1^{\frac{1}{\varepsilon}} \Pr[m_r(e) \geq t]\ \text{d}t \\
    & \leq 1 + \int_1^{\frac{1}{\varepsilon}} \frac{2}{t}\ \text{d}t \\
    & = 1+ 2 \ln \frac{1}{\varepsilon}.
\end{align*}
Combining this with the inequality $\rho_{H,E}^*(S_{r^*}) \leq \sum_{e \in E} \mathbb{E}[m_r(e)] \cdot y(e)$ obtained before, we finally have $\rho_{H,E}^*(S_{r^*}) \leq (1+2\ln \frac{1}{\varepsilon}) \cdot \sum_{e \in E} y(e) = (1+2\ln \frac{1}{\varepsilon}) \cdot \rho_{H,E}^*(x)$.

\paragraph{Analysis based on the degeneracy of the incidence graph of $H$.}
In the second part of our analysis, we bound the rounding gap using the degeneracy of the incidence graph of $H$.
Let $\prec$ be a degeneracy ordering of the incidence graph of $H$.
By definition, for every vertex $v \in V(H)$, there are at most $\mu_H$ hyperedges $e \in E(H)$ satisfying $v \prec e$ and $v \in e$.
On the other hand, for every hyperedge $e \in E(H)$, there are at most $\mu_H$ vertices $v \in V(H)$ satisfying $e \prec v$ and $v \in e$.
For convenience, for each $v \in V(H)$, we write $E_{\prec v}(H) = \{e \in E(H): e \prec v\}$ and $E_{\succ v}(H) = \{e \in E(H): v \prec e\}$.
Recall that $E$ is a given subset of hyperedges of $H$, and $y: E \rightarrow [0,1]$ is a fractional $E$-edge cover of $x$ of the minimum cost.
Let $E_{\prec v} = E_{\prec v}(H) \cap E$ and $E_{\succ v} = E_{\succ v}(H) \cap E$.
Similarly, for each $e \in E(H)$, we write $V_{\prec e}(H) = \{v \in V(H): v \prec e\}$ and $V_{\succ e}(H) = \{v \in V(H): e \prec v\}$.

 As in the first part of the analysis, for every $r \in [0,1]$, we shall define a canonical fractional $E$-edge cover $y_r: E \rightarrow [0,1]$ of $S_r$.
Let $m_r^{\leftarrow}(e) = \max\{\frac{2}{x(v)}: v \in S_r \cap V_{\succ e}(H) \cap e\}$; if $S_r \cap V_{\succ e}(H) \cap e = \emptyset$, simply set $m_r^{\leftarrow}(e) = 0$.
On the other hand, let
\begin{equation*}
    m_r^{\rightarrow}(e) = \left\{
    \begin{array}{ll}
        \frac{1}{y(e)} & \text{if there exists } v \in S_r \cap e \text{ with } y(e) \geq \frac{x(v)}{2 \mu_H}, \\
        0 & \text{otherwise}.
    \end{array}
    \right.
\end{equation*}
Then we define $y_r(e) = \min\{(m_r^{\leftarrow}(e)+m_r^{\rightarrow}(e)) \cdot y(e),1\}$.

\begin{observation}
    $y_r$ is a fractional $E$-edge cover of $S_r$.
\end{observation}
\begin{proof}
We say a vertex $v \in V(H)$ is \emph{left-heavy} if $\sum_{e \in E_{\prec v}: v \in e} y(e) \geq \sum_{e \in E_{\succ v}: v \in e} y(e)$, and is \emph{right-heavy} otherwise.
Observe that $E = E_{\prec v} \cup E_{\succ v}$ for every $v \in V(H)$.
Thus, for each left-heavy vertex $v \in V(H)$, we have $\sum_{e \in E_{\prec v}: v \in e} y(e) \geq \frac{1}{2} \cdot \sum_{e \in E: v \in e} y(e) \geq \frac{x(v)}{2}$.
Similarly, for each right-heavy vertex $v \in V(H)$, we have $\sum_{e \in E_{\succ v}: v \in e} y(e) \geq \frac{x(v)}{2}$.

To show $y_r$ is a fractional $E$-edge cover of $S_r$, we need $\sum_{e \in E: v \in e} y_r(e) \geq 1$ for all $v \in S_r$.
Consider a left-heavy vertex $v \in S_r$.
If there exists $e \in E$ satisfying $v \in e$ and $y_r(e) = 1$, we are done.
So assume $y_r(e) < 1$ for all $e \in E$ with $v \in e$.
Note that for each $e \in E_{\prec v}$ with $v \in e$, we have $v \in S_r \cap V_{\succ e}(H) \cap e$, and thus $m_r^\leftarrow(e) \geq \frac{2}{x(v)}$.
Therefore,
\begin{equation*}
    \sum_{e \in E: v \in e} y_r(e) \geq \sum_{e \in E_{\prec v}: v \in e} y_r(e) \geq \sum_{e \in E_{\prec v}: v \in e} m_r^\leftarrow(e) \cdot y(e) \geq \sum_{e \in E_{\prec v}: v \in e} \frac{2 y(e)}{x(v)}.
\end{equation*}
Since $v$ is left-heavy, we have $\sum_{e \in E_{\prec v}: v \in e} y(e) \geq \frac{x(v)}{2}$ as observed before.
It then follows that $\sum_{e \in E: v \in e} y_r(e) \geq 1$.
Next, consider a right-heavy vertex $v \in S_r$.
We have $\sum_{e \in E_{\succ v}: v \in e} y(e) \geq \frac{x(v)}{2}$.
Note that $|E_{\succ v}| \leq |E_{\succ v}(H)| \leq \mu_H$.
So there exists $e^* \in E_{\succ v}$ with $v \in e^*$ such that $y(e^*) \geq \frac{x(v)}{2 \mu_H}$.
By definition, we have $m_r^\rightarrow(e^*) = \frac{1}{y(e^*)}$, which implies $m_r^\rightarrow(e^*) \cdot y(e^*) = 1$ and hence $y_r(e^*) = 1$.
\end{proof}

Let $r \in [0,1]$ be a random number sampled from the uniform distribution on $[0,1]$.
By the definition of $r^*$ and the above observation, we have
\begin{align*}
    \rho_{H,E}^*(S_{r^*}) &\leq \mathbb{E}[\rho_{H,E}^*(S_r)] \leq \mathbb{E}\left[\sum_{e \in E} y_r(e)\right] \leq \mathbb{E}\left[\sum_{e \in E} (m_r^{\leftarrow}(e)+m_r^{\rightarrow}(e)) \cdot y(e)\right] \\
    &= \sum_{e \in E} (\mathbb{E}[m_r^{\leftarrow}(e) + m_r^{\rightarrow}(e)]) \cdot y(e) = \sum_{e \in E} (\mathbb{E}[m_r^{\leftarrow}(e)] + \mathbb{E}[m_r^{\rightarrow}(e)]) \cdot y(e).
\end{align*}
Next, we bound $\mathbb{E}[m_r^{\leftarrow}(e)]$ and $\mathbb{E}[m_r^{\rightarrow}(e)]$ separately for each hyperedge $e \in E$.
Since $m_r(e) \leq \sum_{v \in S_r \cap V_{\succ e}(H) \cap e} \frac{1}{x(v)}$, we have $\mathbb{E}[m_r^{\leftarrow}(e)] \leq \sum_{v \in V_{\succ e}(H) \cap e} \Pr[v \in S_r] \cdot \frac{1}{x(v)}$.
For all $v \in V(H)$, the length of $I_v$ is $x(v)$ and thus $\Pr[v \in S_r] \leq x(v)$.
Therefore, 
\begin{equation*}
    \mathbb{E}[m_r^{\leftarrow}(e)] \leq |V_{\succ e}(H) \cap e| \leq \mu_H.
\end{equation*}
To bound $\mathbb{E}[m_r^{\rightarrow}(e)]$, observe that $\mathbb{E}[m_r^{\rightarrow}(e)] = \Pr[m_r^{\rightarrow}(e) > 0] \cdot \frac{1}{y(e)}$.
Furthermore, $m_r^{\rightarrow}(e) > 0$ if and only if $S_r$ contains some vertex in the set $U = \{v \in e: x(v) \leq 2 \mu_H \cdot y(e)\}$
Pick a number $q \in \bigcap_{v \in e} I_v$, which exists by Observation~\ref{obs-stab}.
Note that for all $v \in U$, we have $I_v \subseteq [q - 2 \mu_H \cdot y(e),q+2 \mu_H \cdot y(e)]$, because $q \in I_v$ and the length of $I_v$ is $x(v)$.
Thus, $S_r$ contains some vertex $v \in U$ only if $r \in [q - 2 \mu_H \cdot y(e),q+2 \mu_H \cdot y(e)]$, which implies that $\Pr[m_r^{\rightarrow}(e) > 0] \leq 4 \mu_H \cdot y(e)$.
It follows that $\mathbb{E}[m_r^{\rightarrow}(e)] \leq 4 \mu_H$.
Finally, we conclude that
\begin{equation*}
    \rho_{H,E}^*(S_{r^*}) \leq \sum_{e \in E} (\mathbb{E}[m_r^{\leftarrow}(e)] + \mathbb{E}[m_r^{\rightarrow}(e)]) \cdot y(e) \leq 6 \mu_H \cdot \sum_{e \in E} y(e) = 6 \mu_H \cdot \rho_{H,E}^*(x).
\end{equation*}

\paragraph{Putting everything together.}
Now we combine two parts of our analysis to prove Theorem~\ref{theorem:ab-sep-stronger}.
The second part of the analysis shows that we can obtain an $(A,B)$-separator $S = S_{r^*} \subseteq R$ with $\rho_{H,E}^*(S) \leq 6 \mu_H \cdot \rho_H^*(x)$.
To achieve the bound $\rho_{H,E}^*(S) \leq (8+4\ln \alpha_H(R)) \cdot \rho_{H,E}^*(x)$, we recall the first part of the analysis, which shows $\rho_{H,E}^*(S) \leq (1+2\ln \frac{1}{\varepsilon}) \cdot \rho_{H,E}^*(x)$, where $\varepsilon$ is the minimum nonzero value of $x$.
However, this does not directly give us the desired bound, because the minimum nonzero value of $x$ can be arbitrarily small.
Therefore, we will first round $x$ to another fractional $(A,B)$-separator $\tilde{x}$ whose values are in $\{0\} \cup [\varepsilon,1]$ for $\varepsilon = \Omega(\frac{1}{\alpha_H(R)})$, and then apply our rounding algorithm to $\tilde{x}$.

\begin{lemma}
    Given any $\varepsilon < \frac{1}{2 \alpha_H(R)}$, one can compute in $\lVert H \rVert^{O(1)}$ time a fractional $(A,B)$-separator $\tilde{x}: V(H) \rightarrow [0,1]$ with $\supp(\tilde{x}) \subseteq R$ such that
    \begin{itemize}
        \item $\tilde{x}(v) \in \{0\} \cup [\varepsilon,1]$ for all $v \in V(H)$,
        \item $\rho_{H,E}^*(\tilde{x}) \leq \frac{\rho_{H,E}^*(x)}{1 - 2 \varepsilon \alpha_H(R)}$.
    \end{itemize}
\end{lemma}
\begin{proof}
Let $z = 1/(1 - 2 \varepsilon \alpha_H(R))$.
We define $\tilde{x}(v) = 0$ if $x(v) < \varepsilon$ and $\tilde{x}(v) = \min\{z \cdot x(v),1\}$ if $x(v) \geq \varepsilon$.
Clearly, $\tilde{x}(v) \in \{0\} \cup [\varepsilon,1]$ for all $v \in V(H)$.
To see that $\tilde{x}$ is a fractional $(A,B)$-separator, consider a pair $(a,b) \in A \times B$ and an induced path $\pi = (v_0,v_1,\dots,v_\ell)$ in $H$ with $v_0 = a$ and $v_\ell = b$.
We want to show $\sum_{i=0}^\ell \tilde{x}(v_i) \geq 1$.
Let $V_0(\pi) = \{v_i \in \pi: i \text{ is even}\}$ and $V_1(\pi) = \{v_i \in \pi: i \text{ is odd}\}$.
Since $\pi$ is an induced path, $V_0(\pi)$ and $V_1(\pi)$ are both independent sets in $H$, which implies $|V_0(\pi) \cap R| \leq \alpha_H(R)$ and $|V_1(\pi) \cap R| \leq \alpha_H(R)$.
Thus, $|V(\pi) \cap R| \leq 2 \alpha_H(R)$.
Define $V_\mathsf{small} = \{v_i \in \pi: 0 < x(v_i) < \varepsilon\}$.
Observe that $V_\mathsf{small} \subseteq V(\pi) \cap R$, as $x(v) = 0$ for all $v \in V(H) \backslash R$.
It follows that $|V_\mathsf{small}| \leq 2 \alpha_H(R)$.
If $x(v) \geq \frac{1}{z}$ for some $v \in V(\pi) \backslash V_\mathsf{small}$, then $\tilde{x}(v) = \min\{z \cdot x(v),1\} \geq 1$ and we are done.
So assume $x(v) < \frac{1}{z}$ for all $v \in V(\pi) \backslash V_\mathsf{small}$.
In this case, $\tilde{x}(v) = z \cdot x(v)$ for all $v \in V(\pi) \backslash V_\mathsf{small}$ and thus
\begin{equation*}
    \sum_{i=0}^\ell \tilde{x}(v) = z \cdot \left( \sum_{i=0}^\ell x(v_i) - \sum_{v \in V_\mathsf{small}} x(v) \right).
\end{equation*}
Note that $\sum_{v \in V_\mathsf{small}} x(v) < \varepsilon |V_\mathsf{small}| \leq 2 \varepsilon \alpha_H(R)$, and $\sum_{i=0}^\ell x(v_i) \geq 1$ as $x$ is fractional $(A,B)$-separator.
Hence, $\sum_{i=0}^\ell x(v_i) - \sum_{v \in V_\mathsf{small}} x(v) > 1 - 2 \varepsilon \alpha_H(R) = \frac{1}{z}$, which implies $\sum_{i=0}^\ell \tilde{x}(v) > 1$.

To see $\rho_{H,E}^*(\tilde{x}) \leq z \cdot \rho_{H,E}^*(x)$, consider a fractional $E$-edge cover $y: E \rightarrow [0,1]$ of $x$ with cost $\rho_{H,E}^*(x)$.
By definition, we have $\sum_{e \in E: v \in e} y(e) \geq x(v)$ for all $v \in V(H)$.
Define $\tilde{y}: E \rightarrow [0,1]$ as $\tilde{y}(e) = \min\{z \cdot y(e),1\}$ for all $e \in E$.
We claim that $\tilde{y}$ is a fractional $E$-edge cover of $\tilde{x}$.
Consider a vertex $v \in V(H)$ and we want to show $\sum_{e \in E: v \in e} \tilde{y}(e) \geq \tilde{x}(v)$.
If there exists a hyperedge $e \in E$ with $v \in e$ such that $y(e) \geq \frac{1}{z}$, then we are done since $\tilde{y}(e) = 1 \geq \tilde{x}(v)$.
So assume $y(e) < \frac{1}{z}$ for all $e \in E$ with $v \in e$.
In this case, 
\begin{equation*}
\sum_{e \in E: v \in e} \tilde{y}(e) = z \cdot \sum_{e \in E: v \in e} y(e) \geq z \cdot x(v) \geq \tilde{x}(v).
\end{equation*}
As the cost of $y$ is $\rho_{H,E}^*(x)$, the cost of $\tilde{y}$ is at most $z \cdot \rho_{H,E}^*(x)$.
So we conclude that $\rho_{H,E}^*(\tilde{x}) \leq z \cdot \rho_{H,E}^*(x)$.
\end{proof}

\noindent
Using the above lemma with $\varepsilon = \frac{1}{4 \alpha_H(R)}$, we compute a fractional $(A,B)$-separator $\tilde{x}:V(H) \rightarrow [0,1]$ with $\supp(\tilde{x}) \subseteq R$ that satisfies $\tilde{x} \in \{0\} \cup [\frac{1}{4 \alpha_H(R)},1]$ for all $v \in V(H)$ and $\rho_{H,E}^*(\tilde{x}) \leq 2 \rho_{H,E}^*(x)$.
Then we apply our rounding algorithm to $\tilde{x}$ to obtain an $(A,B)$-separator $S \subseteq R$, which satisfies 
\begin{align*}
    \rho_{H,E}^*(S) &\leq \left(1+2 \ln \frac{1}{\varepsilon} \right) \cdot \rho_{H,E}^*(\tilde{x}) = (1+2 \ln 4\alpha_H(R)) \cdot \rho_{H,E}^*(\tilde{x}) \\
    &\leq (2+4 \ln 4\alpha_H(R)) \cdot \rho_{H,E}^*(x) \\
    &\leq (8+4 \ln \alpha_H(R)) \cdot \rho_{H,E}^*(x).
\end{align*}
This completes the proof of Theorem~\ref{theorem:ab-sep-stronger} (and thus Theorem~\ref{theorem:ab-sep-main}).

\section{Computing balanced separators}
\label{section:bal-sep}
In this section, based on the $(A,B)$-separator algorithm presented in the previous section, we give a polynomial-time algorithm for finding a ``balanced'' separator in a hypergraph $H$ with a small fractional edge cover number.
The separator separates a given set $Z \subseteq V(H)$ in a balanced way in terms of the fractional edge cover number.
Formally, we have the following definition.

\begin{definition}\label{def-balanced}
Let $H$ be a hypergraph, $Z \subseteq V(H)$ be a set, and $\varphi \in (0,1)$ be a number.
A \textbf{$(Z,\varphi)$-balanced separator} in $H$ is a subset $S \subseteq V(H)$ such that $\rho_H^*(C \cap Z) \leq \varphi \cdot \rho_H^*(Z)$ for every connected component $C$ of $H-S$.
\end{definition}

Again, our algorithm is based on LP rounding.
Unfortunately, the above definition for balanced separators does not give us a natural way to define fractional balanced separators.
Therefore, we also need to introduce another definition, which defines balanced separators for a function $\gamma: E(H) \rightarrow [0,1]$ instead of a set $Z \subseteq V(H)$.

\begin{definition} \label{def-balanced2}
Let $H$ be a hypergraph, $\gamma:E(H)\rightarrow [0,1]$ be a function, and $\varphi \in (0,1)$ be a number.
A \textbf{$(\gamma,\varphi)$-balanced separator} in $H$ is a subset $S \subseteq V(H)$ such that for every connected component $C$ of $H-S$ it holds that
\begin{equation*}
    \sum_{e \in E(H): e \cap C \neq \emptyset} \gamma(e) \leq \varphi \cdot \sum_{e \in E(H)} \gamma(e).
\end{equation*}
A \textbf{fractional $(\gamma,\varphi)$-balanced separator} in $H$ is a function $x:V(H)\rightarrow [0,1]$ such that for every hyperedge $e \in E(H)$ it holds that
\begin{equation*}
    \sum_{e' \in E(H)} \mathsf{dist}^*_{H,x}(e,e')\cdot \gamma(e') \geq (1-\varphi) \cdot \sum_{e' \in E(H)} \gamma(e').
\end{equation*}
\end{definition}

Thus, for a set $Z\subseteq V(H)$, we have two ways of defining a balanced separator -- Definition~\ref{def-balanced} and Definition~\ref{def-balanced2}.
%
The following observation establishes a connection between the two different ways of defining balanced separators. 

\begin{fact} \label{fact-baldef}
    Let $Z \subseteq V(H)$ be a set, $\eta \in (0,1)$ be a number, and $\gamma: E(H) \rightarrow [0,1]$ be a fractional edge cover of $Z$ satisfying $\sum_{e \in E(H)} \gamma(e) = \rho_H^*(Z)$.
    Then a $(\gamma,\varphi)$-balanced separator in $H$ is also a $(Z,\varphi)$-balanced separator in $H$.
\end{fact}
\begin{proof}
Let $S \subseteq V(H)$ be a $(\gamma,\varphi)$-balanced separator.
For each connected component $C$ of $H-S$,
\begin{equation*}
    \rho_H^*(Z \cap C) \leq \sum_{e \in E(H): e \cap C \neq \emptyset} \gamma(e) \leq \varphi \cdot \sum_{e \in E(H)} \gamma(e) = \varphi \cdot \rho_H^*(Z).
\end{equation*}
Therefore, $S$ is a $(Z,\varphi)$-balanced separator in $H$.
\end{proof}

We also notice that in Definition~\ref{def-balanced2}, a (integral) balanced separator naturally induces a fractional balanced separator (for the same function).
For a set $S \subseteq V(H)$, define $\mathbf{1}_S: V(H) \rightarrow [0,1]$ as $\mathbf{1}_S(v) = 1$ for $v \in S$ and $\mathbf{1}_S(v) = 0$ for $v \in V(H) \backslash S$.

\begin{fact} \label{fact-intisfr}
A set $S \subseteq V(H)$ is a $(\gamma,\varphi)$-balanced separator in $H$ if and only if the function $\mathbf{1}_S$ is a fractional $(\gamma,\varphi)$-balanced separator in $H$.
\end{fact}
\begin{proof}
Observe that for any $e,e' \in E(H)$, we have $\mathsf{dist}^*_{H,\mathbf{1}_S}(e,e') = 0$ if $e \cap C \neq \emptyset$ and $e' \cap C \neq \emptyset$ for some connected component $C$ of $H-S$, and $\mathsf{dist}^*_{H,\mathbf{1}_S}(e,e') = 1$ otherwise.
Also, each hyperedge of $H$ intersects at most one connected component of $H-S$.
Therefore, for each connected component $C$ of $H-S$ and each hyperedge $e \in E(H)$ with $e \cap C \neq \emptyset$, we have
\begin{equation*}
    \sum_{e' \in E(H): e' \cap C \neq \emptyset} \gamma(e') = \sum_{e' \in E(H)} \gamma(e') - \sum_{e' \in E(H)} \mathsf{dist}^*_{H,\mathbf{1}_S}(e,e')\cdot \gamma(e').
\end{equation*}
Assume $\mathbf{1}_S$ is a fractional $(\gamma,\varphi)$-balanced separator.
Consider a connected component $C$ of $H-S$.
Let $e \in E(H)$ be a hyperedge satisfying $e \cap C \neq \emptyset$.
If such a hyperedge does not exist, then $\sum_{e' \in E(H): e' \cap C \neq \emptyset} \gamma(e') = 0$ and we are done.
Otherwise, we have $\sum_{e' \in E(H)} \mathsf{dist}^*_{H,\mathbf{1}_S}(e,e') \cdot \gamma(e') \geq (1-\varphi) \cdot \sum_{e \in E(H)} \gamma(e)$.
By the above equation, $\sum_{e' \in E(H): e' \cap C \neq \emptyset} \gamma(e') \leq \varphi \cdot \sum_{e \in E(H)} \gamma(e)$.
Next, assume $S$ is a $(\gamma,\varphi)$-balanced separator.
Consider an hyperedge $e \in E(H)$.
If $e \subseteq S$, then $\mathsf{dist}^*_{H,\mathbf{1}_S}(e,e') = 1$ for all $e' \in E(H)$ and we are done.
Otherwise, $e$ intersects some connected component $C$ of $H-S$.
We have $\sum_{e' \in E(H): e' \cap C \neq \emptyset} \gamma(e') \leq \varphi \cdot \sum_{e \in E(H)} \gamma(e)$.
The above equation then implies $\sum_{e' \in E(H)} \mathsf{dist}^*_{H,\mathbf{1}_S}(e,e')\cdot \gamma(e') \geq (1-\varphi) \cdot \sum_{e \in E(H)} \gamma(e)$.
\end{proof}

Now we are ready to formally describe the main result of this section.
In fact, our algorithm can not only compute a balanced separator, but also guarantee that the separator is contained in some given set $R \subseteq V(H)$ satisfying certain conditions.
This feature is important in our final algorithms for fractional hypertree width.

\begin{theorem}\label{theorem:bal-separator-main}
Given a hypergraph $H$, a set $Z \subseteq V(H)$ with a fractional edge cover $\gamma:E(H) \rightarrow [0,1]$ of $Z$ satisfying $\sum_{e \in e(H)} \gamma(e) = \rho_H^*(Z)$, and a set $R \subseteq V(H)$ that is a $(\gamma,\frac{1}{2})$-balanced separator in $H$,
one can compute in $\lVert H \rVert^{O(1)}$ time a $(Z,\frac{5}{6})$-balanced separator $S \subseteq R$ in $H$ such that
    \begin{equation*}
        \rho_H^*(S) \leq (\min\{8+4\ln \alpha_H(R),6\mu_H\}+1) \cdot (104 + 16 \log \rho_H^*(S')) \cdot \rho_H^*(S')
    \end{equation*}
for any $(\gamma,\frac{1}{2})$-balanced separator $S' \subseteq R$ in $H$.

%
%
\end{theorem}

On a high level, our algorithm first computes an optimal fractional $(\gamma,\frac{1}{2})$-balanced separator $x:V(H) \rightarrow [0,1]$ in $H$ with $\supp(x) \subseteq R$ using linear programming, and then round $x$ to the desired integral separator $S$ by applying a ball-growing argument together with the $(A,B)$-separator algorithm in Theorem~\ref{theorem:ab-sep-main}.
The main difficulty lies in the rounding step.
Therefore, in what follows, we shall first present our rounding algorithm and then discuss the linear program for computing $x$.

\subsection{Rounding algorithm}

In this section, we show our algorithm for rounding a fractional balanced separator to an integral one, which is the main ingredient for Theorem~\ref{theorem:bal-separator-main}.
Specifically, we prove the following result.

\begin{theorem} \label{thm-ballgrow}
Let $\varphi \in (0,1)$ be a constant.
Given a hypergraph $H$, a function $\gamma: E(H) \rightarrow [0,1]$, and a fractional $(\gamma,\varphi)$-balanced separator $x: V(H) \rightarrow [0,1]$ in $H$ with $\supp(x) = R$, one can compute in $\lVert H \rVert^{O(1)}$ time a $(\gamma,\varphi')$-balanced separator $S \subseteq R$ for $\varphi' = \frac{1+3\varphi}{2+2\varphi}$ such that
    \begin{equation*}
        \rho_H^*(S) \leq (\min\{8+4\ln \alpha_H(R),6\mu_H\}+1) \cdot \frac{44 + 8 \log (\rho_H^*(x)/(1-\varphi))}{1-\varphi} \cdot \rho_H^*(x).
    \end{equation*}
\end{theorem}

We begin with introducing several notations.
Let $Q \subseteq V(H)$ be a subset of vertices.
We define a function $x \Cap Q: V(H) \rightarrow [0,1]$ which maps every $v \in Q$ to $x(v)$ and maps every $v \in V(H) \backslash Q$ to $0$.
Write $\gamma(Q) = \sum_{e \in E(H): e \cap Q \neq \emptyset} \gamma(e)$.
For a vertex $v \in Q$, define $\mathcal{B}_c^Q(r)$ as the ``$r$-ball'' in $H[Q]$ centered at $c$ with respect to the $x$-weighted distance, i.e., $\mathcal{B}_c^Q(r) = \{v \in Q: \mathsf{dist}^*_{H[Q],x}(c,v) \leq r\}$.
We need the following two important observations for the balls $\mathcal{B}_c^Q(r)$.
Denote $\mathcal{B}_c^{V(H)}(r)$ by $\mathcal{B}_c(r)$ for simplicity.
The first observation is that, when $r$ is a sufficiently small constant $r\in [0,1)$, the ``volume'' of the ball $\cB^Q_c(r)$, i.e., $\gamma(\cB^Q_c(r))$ is only a small fraction of $\gamma(V(H))$.

\begin{observation} \label{obs-smallball}
$\gamma(\cB^Q_c(r))\leq \frac{\varphi}{1-r} \cdot \gamma(V(H))$ for any $Q \subseteq V(H)$, $c\in Q$, and $r\in [0,1)$.
\end{observation}
\begin{proof}
Note that $\cB^Q_c(r) \subseteq \cB_c(r)$ for any $Q \subseteq V(H)$.
So we can assume $Q = V(H)$, i.e., it suffices to show $\gamma(\cB_c(r))\leq \frac{\varphi}{1-r} \cdot \gamma(V(H))$ for any $c \in Q$ and $r \in [0,1)$.
Let $e_c \in E(H)$ be a hyperedge with $c \in e_c$.
Since $x$ is a fractional $(\gamma,\varphi)$-separator, $\sum_{e \in E(H)} \mathsf{dist}^*_{H,x}(e_c,e) \cdot \gamma(e) \geq (1-\varphi) \cdot \gamma(V(H))$.
Let $E_0 = \{e \in E(H): e \cap \cB_c(r) \neq \emptyset\}$.
Note that for all $e \in E_0$, we have $\mathsf{dist}^*_{H,x}(e_c,e) \leq r$ and hence $\mathsf{dist}^*_{H,x}(e_c,e) \leq r$.
It follows that
\begin{equation*}
    (1-\varphi) \cdot \gamma(V(H)) \leq \sum_{e \in E(H)} \mathsf{dist}^*_{H,x}(e_c,e) \cdot \gamma(e) \leq r \sum_{e \in E_0} \gamma(e) + \sum_{e \in E(H) \backslash E_0} \gamma(e).
\end{equation*}
By definition, $\sum_{e \in E_0} \gamma(e) = \gamma(\cB_c(r))$ and $\sum_{e \in E(H) \backslash E_0} \gamma(e) = \gamma(V(H)) - \gamma(\cB_c(r))$.
The above inequality then implies that $(1-\varphi) \cdot \gamma(V(H)) \leq \gamma(V(H)) - (1-r) \cdot \gamma(\cB_c(r))$.
Equivalently, we have $\gamma(\cB_c(r))\leq \frac{\varphi}{1-r} \cdot \gamma(V(H))$.
\end{proof}

\noindent
The second observation is that for $r,r' \in [0,1]$ with $r' > r$, one can obtain a fractional $(\mathcal{B}_c^Q(r),Q \backslash \mathcal{B}_c^Q(r'))$-separator by taking the vertices in $\mathcal{B}_c^Q(r') \backslash \mathcal{B}_c^Q(r)$ with the weights $x(\cdot)$ amplified by a certain factor.

\begin{observation} \label{obs-stripsep}
    Let $Q \subseteq V(H)$, $q = \max_{v \in Q} x(v)$, and $c \in Q$.
    For any non-negative numbers $r,r' \in [0,1]$ satisfying $q < r'-r$, the function $\tilde{x}: Q \rightarrow [0,1]$ defined as $\tilde{x}(v) = x(v)/(r'-r-q)$ for $v \in \mathcal{B}_c^Q(r') \backslash \mathcal{B}_c^Q(r)$ and $\tilde{x}(v) = 0$ otherwise is a fractional $(\mathcal{B}_c^Q(r),Q \backslash \mathcal{B}_c^Q(r'))$-separator in $H[Q]$.
\end{observation}
\begin{proof}
Consider a path $(v_0,v_1,\dots,v_\ell)$ in $H[Q]$ where $v_0 \in \mathcal{B}_c^Q(r)$ and $v_\ell \in Q \backslash \mathcal{B}_c^Q(r')$.
We need to show $\sum_{i=0}^\ell \tilde{x}(v_i) \geq 1$.
Let $\alpha = \max\{i \in \{0\} \cup [\ell]: v_i \in \mathcal{B}_c^Q(r)\}$ and $\beta = \min\{i \in [\ell] \backslash [\alpha]: v_i \in Q\setminus \mathcal{B}_c^Q(r')\}$.
Note that $\alpha$ and $\beta$ both exist since $v_0 \in \mathcal{B}_c^Q(r)$ and $v_\ell \in Q \backslash \mathcal{B}_c^Q(r')$.
By definition, we have $v_{\alpha+1},\dots,v_{\beta-1} \notin \mathcal{B}_c^Q(r) \cup (Q \backslash \mathcal{B}_c^Q(r'))$ and thus $v_{\alpha+1},\dots,v_{\beta-1} \in \mathcal{B}_c^Q(r') \backslash \mathcal{B}_c^Q(r)$.
Also, $\mathsf{dist}^*_{H[Q],x}(c,v_\alpha) \leq r$ and $\mathsf{dist}^*_{H[Q],x}(c,v_\beta) > r'$, which implies that $\sum_{i=\alpha+1}^\beta x(v_i) > r'-r$ and hence $\sum_{i=\alpha+1}^{\beta-1} x(v_i) > r'-r-q$.
Since $v_{\alpha+1},\dots,v_{\beta-1} \in \mathcal{B}_c^Q(r') \backslash \mathcal{B}_c^Q(r)$, we have
\begin{equation*}
    \sum_{i=0}^\ell \tilde{x}(v_i) \geq \sum_{i=\alpha+1}^{\beta-1} \tilde{x}(v_i) = \sum_{i=\alpha+1}^{\beta-1} \frac{x(v_i)}{r'-r-q} > 1,
\end{equation*}
which completes the proof.
\end{proof}

Our rounding algorithm computes the desired separator $S$ iteratively, using Observations~\ref{obs-smallball} and~\ref{obs-stripsep}.
In each iteration, we cut off a small ``piece'' from the hypergraph by removing some vertices.
At the end, all the removed vertices form the separator $S$.
A single cut-off step is achieved by the following lemma.

\begin{lemma} \label{lem-onestepballgrow}
Let $r,\phi \in (0,1)$ be constants.
Given a hypergraph $H$, a function $\gamma: E(H) \rightarrow [0,1]$, a fractional $(\gamma,\varphi)$-balanced separator $x: V(H) \rightarrow [0,1]$ in $H$ with $\supp(x) = R$, and a set $Q \subseteq V(H)$ satisfying $x(v) < \frac{r}{18+ 4 \log (\rho_H^*(x)/r)}$ for all $v \in Q$, one can compute in $\lVert H \rVert^{O(1)}$ time two disjoint sets $P,S \subseteq Q$ satisfying the following.
\begin{enumerate}[label=(\roman*)]
    \item $P \neq \emptyset$, $S \subseteq R$, and $S$ is a $(P,Q \backslash P)$-separator in $H[Q]$.
    \item $\gamma(P) \leq \frac{\varphi}{1-r} \cdot \gamma(V(H))$.
    \item $\rho_H^*(S) \leq \min\{8+4\ln \alpha_H(R),6\mu_H\} \cdot \frac{18 + 4 \log (\rho_H^*(x)/r)}{r} \cdot \rho_H^*(x \Cap P)$.
\end{enumerate}
\end{lemma}
\begin{proof}
Pick an arbitrary vertex $c \in Q$.
Set $\delta = \frac{r}{9+2 \log (\rho_H^*(x)/r)}$ and $q = \max_{v \in Q} x(v)$.
We have $q < \frac{\delta}{2}$ by assumption.
We write $r_i = \frac{r}{2} + i \delta$ and $L_i = \cB_c^Q(r_{i+1}) \backslash \cB_c^Q(r_i)$ for $i \in \mathbb{Z}$.
Observe that for any two vertices $u \in \cB_c^Q(r_{i-1})$ and $v \in L_i$, there is no hyperedge of $H$ containing both $u$ and $v$.
Indeed, if there exists a hyperedge $e \in E(H)$ such that $u,v \in e$, then $\mathsf{dist}^*_{H[Q],x}(c,v) \leq \mathsf{dist}^*_{H[Q],x}(c,u) + x(v) \leq r_{i-1} + q < r_{i-1} + \delta = r_i$, which contradicts the fact that $v \in L_i \subseteq Q \backslash \cB_c^Q(r_i)$.

We first show the existence of an index $i^* \in \mathbb{N}$ with $i^* \leq 9+2 \log (\rho_H^*(x)/r)$ such that $\rho_H^*(x \Cap L_{i^*}) \leq \rho_H^*(x \Cap \cB_c^Q(r_{i^*-1}))$.
If $\cB_c^Q(r_0) = \cB_c^Q(r_2)$, then $L_1 = \emptyset$ and we can simply choose $i^* = 1$.
So suppose $\cB_c^Q(r_0) \subsetneq \cB_c^Q(r_2)$ and let $u \in \cB_c^Q(r_2) \backslash \cB_c^Q(r_0)$.
Consider an induced path $\pi = (v_0,v_1,\dots,v_\ell)$ from $c$ to $u$ in $H[Q]$ satisfying $\sum_{i=0}^\ell x(v_i) = \mathsf{dist}_{H[Q],x}(c,u)$, where $v_0 = c$ and $v_\ell = u$.
Let $j \in \{0\} \cup [\ell]$ be the largest index such that $v_j \in \cB_c^Q(r_0 - q)$.
Note that $j$ must exist since
\begin{equation*}
    \mathsf{dist}^*_{H[Q],x}(c,v_0) = x(c) < q < \frac{r}{2} - q = r_0 - q,
\end{equation*}
and thus $v_0 \in \cB_c^Q(r_0-q)$.
On the other hand, $j < \ell$, because $v_\ell = v \notin \cB_c^Q(r_0)$.
As $v_{j+1} \notin \cB_c^Q(r_0 - q)$, we have $\sum_{i=0}^{j+1} x(v_i) \geq \mathsf{dist}^*_{H[Q],x}(c,v_{j+1}) > r_0 - q$, which implies
\begin{equation*}
    \sum_{i=0}^j x(v_i) \geq r_0 - q - x(v_{j+1}) \geq \frac{r}{2} - 2q \geq \frac{r}{4}.
\end{equation*}
Let $V_0 = \{v_i: i \in \{0\} \cup [j]\text{ is even}\}$ and $V_1 = \{v_i: i \in \{0\} \cup [j]\text{ is odd}\}$.
Since $\pi$ is an induced path, $V_0$ and $V_1$ are independent sets in $H[Q]$, which implies that $\rho_H^*(x \Cap V_0) \geq \sum_{v \in V_0} x(v)$ and $\rho_H^*(x \Cap V_1) \geq \sum_{v \in V_1} x(v)$.
It follows that 
\begin{align*}
    \rho_H^*(x \Cap \{v_0,v_1,\dots,v_j\}) &\geq \max\{\rho_H^*(x \Cap V_0),\rho_H^*(x \Cap V_1)\} \\
    &\geq \max \left\{\sum_{v \in V_0} x(v),\sum_{v \in V_1} x(v)\right\} \\
    &\geq \frac{1}{2} \sum_{i=0}^j x(v_i) \geq \frac{r}{8}.
\end{align*}
In particular, $\rho_H^*(x \Cap \mathcal{B}_c^Q(r_0)) \geq \frac{r}{8}$.
Now we are ready to show the existence of $i^*$.
Assume $\rho_H^*(x \Cap L_i) > \rho_H^*(x \Cap \cB_c^Q(r_{i-1}))$ for all $i \leq 9+2 \log (\rho_H^*(x)/r)$.
As observed at the beginning of the proof, there is no hyperedge containing a vertex in $\cB_c^Q(r_{i-1})$ and another vertex in $L_i$.
Thus,
\begin{equation*}
    \rho_H^*(x \Cap \cB_c^Q(r_{i+1})) \geq \rho_H^*(x \Cap (\cB_c^Q(r_{i-1}) \cup L_i)) = \rho_H^*(x \Cap \cB_c^Q(r_{i-1})) + \rho_H^*(x \Cap L_i).
\end{equation*}
For an index $i \leq 9+2 \log (\rho_H^*(x)/r)$, this implies $\rho_H^*(x \Cap \cB_c^Q(r_{i+1})) > 2 \rho_H^*(x \Cap \cB_c^Q(r_{i-1}))$, simply because $\rho_H^*(x \Cap L_{i+1}) > \rho_H^*(x \Cap \cB_c^Q(r_i))$.
As $\rho_H^*(x \Cap \cB_c^Q(r_0)) \geq \frac{r}{8}$, we have $\rho_H^*(x \Cap \cB_c^Q(r_{i+1})) > 2^{(i+1)/2} \cdot \frac{r}{8}$ for every even index $i \leq 9+2 \log (\rho_H^*(x)/r)$.
If we set $i = 7+2 \lceil \log (\rho_H^*(x)/r) \rceil < 9+2 \log (\rho_H^*(x)/r)$, then $\rho_H^*(x \Cap \cB_c^Q(r_{i+1})) > \rho_H^*(x)$, which is a contradiction.
So the desired index $i^*$ exists.

Next, we discuss our algorithm for computing the sets $P,S \subseteq Q$.
We first find the index $i^*$ described above, which can be done in $\lVert H \rVert^{O(1)}$ time.
Define a function $\tilde{x}: Q \rightarrow [0,1]$ as $\tilde{x}(v) = x(v)/(\delta-q)$ for $v \in L_{i^*}$ and $\tilde{x}(v) = 0$ otherwise.
Clearly, $\supp(\tilde{x}) \subseteq L_{i^*}$.
Also, $\supp(\tilde{x}) \subseteq \supp(x) = R$.
By Observation~\ref{obs-stripsep}, $\tilde{x}$ is a fractional $(\cB_c^Q(r_{i^*}),Q \backslash \cB_c^Q(r_{i^*+1}))$-separator in $H[Q]$.
Therefore, we are able to apply the algorithm in Theorem~\ref{theorem:ab-sep-main} to compute a $(\cB_c^Q(r_{i^*}),Q \backslash \cB_c^Q(r_{i^*+1}))$-separator $S \subseteq L_i^* \cap R$ in $H[Q]$ satisfying that
\begin{align*}
    \rho_H^*(S) &\leq \min\{8+4\ln \alpha_H(R),6\mu_H\} \cdot \rho_H^*(\tilde{x}) \\
    &= \min\{8+4\ln \alpha_H(R),6\mu_H\} \cdot \frac{\rho_H^*(x \Cap L_{i^*})}{\delta - q} \\
    &\leq \min\{8+4\ln \alpha_H(R),6\mu_H\} \cdot \frac{18 + 4 \log (\rho_H^*(x)/r)}{r} \cdot \rho_H^*(x \Cap \cB_c^Q(r_{i^*-1})).
\end{align*}
Finally, we define $P$ as the set of vertices in the connected components of $H[Q \backslash S]$ that intersect $\cB_c^Q(r_{i^*})$.
By construction, $S$ is a $(P,Q \backslash P)$-separator in $H[Q]$.
Furthermore, $c \in \cB_c^Q(r_{i^*-1}) \subseteq P$, which implies $P \neq \emptyset$ and $\rho_H^*(S) \leq \min\{8+4\ln \alpha_H(R),6\mu_H\} \cdot \frac{18 + 4 \log (\rho_H^*(x)/r)}{r} \cdot \rho_H^*(x \Cap P)$.
Therefore, to prove the lemma, it suffices to show $\gamma(P) \leq \frac{\varphi}{1-r} \cdot \gamma(V(H))$.
Since $S$ is a $(\cB_c^Q(r_{i^*}),Q \backslash \cB_c^Q(r_{i^*+1}))$-separator, we have $P \subseteq \cB_c^Q(r_{i^*+1})$.
As $i^* \leq 9+2 \log (\rho_H^*(x)/r)$, $r_{i^*+1} \leq r$ and thus $P \subseteq \cB_c^Q(r)$.
By Observation~\ref{obs-smallball}, we then have $\gamma(P) \leq \gamma(\cB_c^Q(r)) \leq \frac{\varphi}{1-r} \cdot \gamma(V(H))$.
\end{proof}

Based on the above lemma, we are ready to present our algorithm of Theorem~\ref{thm-ballgrow}.
Set $r = \frac{1-\varphi}{2}$.
We denote the algorithm of Lemma~\ref{lem-onestepballgrow} as $\textsc{CutOff}(Q)$, which takes as input a set $Q \subseteq V(H)$ satisfying $x(v) < \frac{r}{18+ 4 \log (\rho_H^*(x)/r)}$ for all $v \in Q$, and returns a pair $(P,S)$ of disjoint subsets of $Q$ satisfying the properties in the lemma.
Recall that Theorem~\ref{thm-ballgrow} requires us to compute a $(\gamma,\varphi')$-balanced separator $S$ in $H$.
Our algorithm for computing $S$ is presented in Algorithm~\ref{alg-balsep}.
Initially, $S = \{v \in V(H): x(v) \geq \frac{r}{18+ 4 \log (\rho_H^*(x)/r)}\}$ (line~2).
Also, we set $P = \emptyset$ and $Q = V(H) \backslash S$ (line~3).
We keep adding vertices to $S$ using a while-loop (line~5-10).
In the $i$-th iteration, we call the sub-routing $\textsc{CutOff}(Q)$ to obtain a pair $(P_i,S_i)$ of disjoint subsets of $Q$ (line~7).
We then add the vertices in $P_i$ to $P$ (line~8), add the vertices in $S_i$ to $S$ (line~9), and remove all vertices in $P_i \cup S_i$ from $Q$ (line~10).
The while-loop terminates when $\gamma(Q) \geq \frac{1-r+\varphi}{2-2r} \cdot \gamma(V(H)) = \frac{1+3\phi}{2+2\phi} \cdot \gamma(V(H))$.
After this, the algorithm returns the final $S$ as its output (line~11).

\begin{algorithm}
    \caption{\textsc{BalancedSeparator}$(H,x,\gamma)$}
    \begin{algorithmic}[1]
        \State $r \leftarrow \frac{1-\varphi}{2}$
        \State $S \leftarrow \{v \in V(H): x(v) \geq \frac{r}{18+ 4 \log (\rho_H^*(x)/r)}\}$
        \State $P \leftarrow \emptyset$ and $Q \leftarrow V(H) \backslash S$
        \State $i \leftarrow 0$
        \While{$\gamma(Q) \geq \frac{1-r+\varphi}{2-2r} \cdot \gamma(V(H))$}
            \State $i \leftarrow i+1$
            \State $(P_i,S_i) \leftarrow \textsc{CutOff}(Q)$
            \State $P \leftarrow P \cup P_i$
            \State $S \leftarrow S \cup S_i$
            \State $Q \leftarrow Q \backslash (P_i \cup S_i)$
        \EndWhile
        \State \textbf{return} $S$
    \end{algorithmic}
    \label{alg-balsep}
\end{algorithm}

We now analyze Algorithm~\ref{alg-balsep}.
First, observe that the while-loop always terminates in at most $n = |V(H)|$ iterations.
Indeed, Lemma~\ref{lem-onestepballgrow} guarantees that the set $P_i$ returned by $\textsc{CutOff}(Q)$ is nonempty and thus $|P|$ increases in each iteration.
Thus, the entire algorithm takes $\lVert H \rVert^{O(1)}$ time.
Define $\varphi' = \frac{1-r+\varphi}{2-2r} = \frac{1+3\varphi}{2+2\varphi}$.
We need to prove that $S$ is a $(\gamma,\varphi')$-balanced separator and $\rho_H^*(S)$ satisfies the desired bound.
Suppose the while-loop has $k$ iterations.
Denote by $Q_i$ the set $Q$ at the end of the $i$-th iteration of the while-loop, for $i \in [k]$.
\begin{observation} \label{obs-partition}
    For every $i \in \{0\} \cup [k]$, $\{\bigcup_{j=1}^i P_j, \bigcup_{j=0}^i S_j, Q_i\}$ is a partition of $V(H)$ and $\bigcup_{j=0}^i S_j$ is a $(\bigcup_{j=1}^i P_j,Q_i)$-separator in $H$.
\end{observation}
\begin{proof}
We apply induction on $i$.
When $i = 0$, the statement holds.
Assume $\{\bigcup_{j=1}^{i-1} P_j, \bigcup_{j=0}^{i-1} S_j, Q_{i-1}\}$ is a partition of $V(H)$ and $\bigcup_{j=0}^{i-1} S_j$ is a $(\bigcup_{j=1}^{i-1} P_j,Q_{i-1})$-separator in $H$.
By construction $P_i$ and $S_i$ are disjoint subsets of $Q_{i-1}$.
Thus, $\{\bigcup_{j=1}^i P_j, \bigcup_{j=0}^i S_j, Q_i\}$ is also a partition of $V(H)$.
Furthermore, $S_i$ is a $(P_i,Q_{i-1} \backslash P_i)$-separator in $H[Q_{i-1}]$, and hence a $(P_i,Q_i)$-separator in $H[Q_{i-1}]$.
This implies $\bigcup_{j=0}^i S_j$ is a $(\bigcup_{j=1}^i P_j,Q_i)$-separator in $H$, since $\bigcup_{j=0}^{i-1} S_j$ is a $(\bigcup_{j=1}^{i-1} P_j,Q_{i-1})$-separator by our induction hypothesis.
\end{proof}

\begin{observation}
    At the end of the algorithm, $\gamma(P) \leq \varphi' \cdot \gamma(V(H))$ and $\gamma(Q) \leq \varphi' \cdot \gamma(V(H))$.
\end{observation}
\begin{proof}
Observe that $\gamma(Q) = \gamma(Q_k) < \varphi' \cdot \gamma(V(H))$, since the while-loop terminates after the $k$-th iteration.
On the other hand, $\gamma(Q_{k-1}) \geq \varphi' \cdot \gamma(V(H))$, since the while-loop does not terminate after the $(k-1)$-th iteration.
By Observation~\ref{obs-partition}, $\bigcup_{j=0}^{k-1} S_j$ is a $(\bigcup_{j=1}^{k-1} P_j,Q_{k-1})$-separator and $\bigcup_{j=0}^{k-1} S_j$ is disjoint from $\bigcup_{j=1}^{k-1} P_j$ and $Q_{k-1}$.
Therefore, there is no hyperedge in $H$ that contains vertices in both $\bigcup_{j=1}^{k-1} P_j$ and $Q_{k-1}$, which implies $\gamma(\bigcup_{j=1}^{k-1} P_j) + \gamma(Q_{k-1}) \leq \gamma(V(H))$ and thus $\gamma(\bigcup_{j=1}^{k-1} P_j) < (1-\varphi') \cdot \gamma(V(H))$.
Since $\gamma(P_k) \leq \frac{\varphi}{1-r} \cdot \gamma(V(H))$ by condition (ii) of Lemma~\ref{lem-onestepballgrow},
\begin{equation*}
    \gamma(P) \leq \gamma\left(\bigcup_{j=1}^{k-1} P_j\right) + \gamma(P_k) < (1-\varphi') \cdot \gamma(V(H)) + \frac{\varphi}{1-r} \cdot \gamma(V(H)) = \varphi' \cdot \gamma(V(H)),
\end{equation*}
which completes the proof.
\end{proof}

The above observation implies that $S$ is a $(\gamma,\varphi')$-balanced separator in $H$.
In what follows, we establish the bound for $\rho_H^*(S)$.

\begin{observation} \label{obs-independent}
    For every hyperedge $e \in E(H)$, there exists at most one index $i \in [k]$ such that $e \cap P_i \neq \emptyset$.
    In particular, we have $\rho_H^*(x) \geq \sum_{i=1}^k \rho_H^*(x \Cap P_i)$.
\end{observation}
\begin{proof}
Assume there exist $i,i' \in [k]$ with $i<i'$ such that $e \cap P_i \neq \emptyset$ and $e \cap P_{i'} \neq \emptyset$.
We have $P_{i'} \subseteq Q_i$.
By Observation~\ref{obs-partition}, $\bigcup_{j=1}^i S_j$ is a $(\bigcup_{j=1}^i P_j,Q_i)$-separator, and in particular a $(P_i,P_{i'})$-separator, in $H[Q_i]$, which is disjoint from $P_i$ and $P_{i'}$.
This contradicts the fact that $e$ contains vertices in both $P_i$ and $P_{i'}$.
\end{proof}

\begin{observation}
    At the end of the algorithm, we have
    \begin{equation*}
        \rho_H^*(S) \leq (\min\{8+4\ln \alpha_H(R),6\mu_H\}+1) \cdot \frac{44 + 8 \log (\rho_H^*(x)/(1-\varphi))}{1-\varphi} \cdot \rho_H^*(x).
    \end{equation*}
\end{observation}
\begin{proof}
Set $t = \frac{18+ 4 \log (\rho_H^*(x)/r)}{r} = \frac{44 + 8 \log (\rho_H^*(x)/(1-\varphi))}{1-\varphi}$, and define $S_0 = \{v \in V(H): x(v) \geq \frac{1}{t}\}$.
Then we have $S = \bigcup_{i=0}^k S_i$.
Note that $\rho_H^*(S_0) \leq t \cdot \rho_H^*(x \Cap S_0) \leq t \cdot \rho_H^*(x)$.
Furthermore, by condition (iii) of Lemma~\ref{lem-onestepballgrow}, $\rho_H^*(S_i) \leq \min\{8+4\ln \alpha_H(R),6\mu_H\} \cdot t \cdot \rho_H^*(x \Cap P_i)$ for every $i \in [k]$.
Combining this with Observation~\ref{obs-independent}, we have $\sum_{i=1}^k \rho_H^*(S_i) \leq \min\{8+4\ln \alpha_H(R),6\mu_H\} \cdot t \cdot \rho_H^*(x)$.
As a result, $\rho_H^*(S) \leq \sum_{i=0}^k \rho_H^*(S_i) \leq (\min\{8+4\ln \alpha_H(R),6\mu_H\}+1) \cdot t \cdot \rho_H^*(x)$.
\end{proof}

Finally, it is clear that the entire algorithm runs in $\lVert H \rVert^{O(1)}$ time, as the algorithm of Lemma~\ref{lem-onestepballgrow} takes $\lVert H \rVert^{O(1)}$ time.
Also, we have $S \subseteq R$ because $S_0 \subseteq R$ by construction and $S_i \subseteq R$ for $i \in [k]$ by condition (i) of Lemma~\ref{lem-onestepballgrow}.
This completes the proof of Theorem~\ref{thm-ballgrow}.

\subsection{Linear program for balanced separator}

In this section, we show that one can compute an optimal fractional balanced separator by solving a linear program (LP).
Specifically, we prove the following lemma.

\begin{lemma} \label{lem-balsepLP}
Given a hypergraph $H$, a function $\gamma:E(H)\rightarrow [0,1]$, and a set $R \subseteq V(H)$ such that $R$ is a $(\gamma,\varphi)$-balanced separator in $H$, one can compute in $\lVert H \rVert^{O(1)}$ time a fractional $(\gamma,\varphi)$-balanced separator $x:V(H)\rightarrow [0,1]$ with $\supp(x)\subseteq R$ that minimizes $\rho_H^*(x)$.
\end{lemma}

The variables of our LP are defined as follows.
For each vertex $v \in V(H)$, we introduce a variable $x_v \in [0,1]$ indicating the value of $x(v)$ for the desired separator $x$.
For each edge $e \in E(H)$, we introduce a variable $y_e \in [0,1]$; these variables together indicate a fractional edge cover $y: E(H) \rightarrow [0,1]$.
In order to guarantee that $x$ is a $(\gamma,\varphi)$-balanced separator, we need additional variables to indicate the distances between vertices and between edges.
For every two vertices $v,v' \in V(H)$, we introduce a variable $d_{v,v'} \in [0,1]$ indicating the distance $\mathsf{dist}^*_{H,x}(v,v')$.
For every two hyperedges $e,e' \in E(H)$, we introduce a variable $d_{e,e'} \in [0,1]$ indicating the distance $\mathsf{dist}^*_{H,x}(e,e')$.
Our LP is formulated as

\begin{equation*}
    \begin{array}{rl}
        & \min \text{ } \sum_{e \in E(H)} y_e  \\[2ex]
        \text{s.t. } & \text{all variables $\in [0,1]$,} \\[1ex]
        & x_v = 0 \text{ for all $v \in V(H) \backslash R$,} \\[1ex]
        & \sum_{e \in E(H): v \in e} y_e \geq x_v \text{ for all $v \in V(H)$,} \\[1ex]
        & \sum_{e' \in E(H)} d_{e,e'} \cdot \gamma(e') \geq (1-\varphi) \cdot \sum_{e' \in E(H)} \gamma(e') \text{ for all $e \in E(H)$,} \\[1ex]
        & d_{v,v} = x_v \text{ for all $v \in V(H)$,} \\[1ex]
        & d_{v,v''} \leq d_{v,v'} + x_{v''} \text{ for all $v \in V(H)$ and $(v',v'') \in E(\gf{H})$,} \\[1ex]
        & d_{e,e'} \leq d_{v,v'} \text{ for all $e,e' \in E(H)$, $v \in e$, $v' \in e'$.}
    \end{array}
\end{equation*}
\\ 
The constraints $x_v = 0$ for $v \in V(H) \backslash R$ guarantee $\supp(x)\subseteq R$, while the constraints $\sum_{e \in E(H): v \in e} y_e \geq 1$ for $e \in E(H)$ guarantee that the variables $y_e$'s represent a fractional edge cover of $x$.
Furthermore, the constraints $\sum_{e' \in E(H)} d_{e,e'} \cdot \gamma(e') \geq (1-\varphi) \cdot \sum_{e' \in E(H)} \gamma(e')$ for $e \in E(H)$ force $x$ to be a fractional $(\gamma,\varphi)$-balanced separator.
The last three constraints are requirements that the distance function $\mathsf{dist}^*_{H,x}$ has to satisfy.
These constraints cannot guarantee that the distance variables are exactly equal to the values of $\mathsf{dist}^*_{H,x}$.
However, they are already sufficient for guaranteeing that $x$ is a fractional $(\gamma,\varphi)$-balanced separator.
\begin{observation} \label{obs-distLB}
    Let $\{\hat{x}_v,\hat{y}_e,\hat{d}_{v,v'},\hat{d}_{e,e'}\}$ be a feasible solution of the LP.
    Define $x: V(H) \rightarrow [0,1]$ as $x(v) = \hat{x}_v$ for all $v \in V(H)$.
    Then $\hat{d}_{v,v'} \leq \mathsf{dist}^*_{H,x}(v,v')$ for all $v,v' \in V(H)$ and $\hat{d}_{e,e'} \leq \mathsf{dist}^*_{H,x}(e,e')$ for all $e,e' \in E(H)$.
    In particular, $x$ is a fractional $(\gamma,\varphi)$-balanced separator in $H$.
\end{observation}
\begin{proof}
For each $v \in V(H)$, we have $\hat{d}_{v,v} = \hat{x}_v = \mathsf{dist}^*_{H,x}(v,v)$, because of the constraint $d_{v,v} = x_v$.
Assume $\hat{d}_{v,v'} > \mathsf{dist}^*_{H,x}(v,v')$ for some $v,v' \in V(H)$.
Consider a path $(v_0,v_1,\dots,v_\ell)$ in $H$ with $v_0 = v$ and $v_\ell = v'$ satisfying that $\sum_{j=0}^\ell x(v_j) = \mathsf{dist}^*_{H,x}(v,v')$.
Note that $\sum_{j=0}^i x(v_j) = \mathsf{dist}^*_{H,x}(v,v_i)$ for all $i \in \{0\} \cup [\ell]$, which implies $x(v_i) = \mathsf{dist}^*_{H,x}(v,v_i) - \mathsf{dist}^*_{H,x}(v,v_{i-1})$.
Let $k \in \{0\} \cup [\ell]$ be the smallest index such that $\hat{d}_{v,v_k} > \mathsf{dist}^*_{H,x}(v,v_k)$.
Clearly, such an index $k$ exists and $k > 0$, because $\hat{d}_{v,v_\ell} = \hat{d}_{v,v'} > \mathsf{dist}^*_{H,x}(v,v') > \mathsf{dist}^*_{H,x}(v,v_\ell)$ and $\hat{d}_{v,v_0} = d_{v,v} = \mathsf{dist}^*_{H,x}(v,v) = \mathsf{dist}^*_{H,x}(v,v_0)$.
Now 
\begin{equation*}
    \hat{d}_{v,v_k} - \hat{d}_{v,v_{k-1}} > \mathsf{dist}^*_{H,x}(v,v_k) - \mathsf{dist}^*_{H,x}(v,v_{k-1}) = x(v_k) = \hat{x}_k,
\end{equation*}
contradicting the constraint $d_{v,v_k} \leq d_{v,v_{k-1}} + x_{v_k}$.
Thus, $\hat{d}_{v,v'} \leq \mathsf{dist}^*_{H,x}(v,v')$ for all $v,v' \in V(H)$.

For two hyperedges $e,e' \in E(H)$, the constraints $d_{e,e'} \leq d_{v,v'}$ for all $v \in e$ and $v' \in e'$ guarantee that $\hat{d}_{e,e'} \leq \min_{v \in e, v' \in e'} \hat{d}_{v,v'}$.
Because of the inequality $\hat{d}_{v,v'} \leq \mathsf{dist}^*_{H,x}(v,v')$ for all $v,v' \in V(H)$, we then have $\hat{d}_{e,e'} \leq \min_{v \in e, v' \in e'} \mathsf{dist}^*_{H,x}(v,v') = \mathsf{dist}^*_{H,x}(e,e')$.

To see $x$ is a fractional $(\gamma,\varphi)$-balanced separator, observe that $\sum_{e' \in E(H)} \hat{d}_{e,e'} \cdot \gamma(e') \geq (1-\varphi) \cdot \sum_{e' \in E(H)} \gamma(e')$ for all $e \in E(H)$, by the corresponding constraints of the LP.
As $\hat{d}_{e,e'} \leq \mathsf{dist}^*_{H,x}(e,e')$, we have $\sum_{e' \in E(H)} \mathsf{dist}^*_{H,x}(e,e') \cdot \gamma(e') \geq (1-\varphi) \cdot \sum_{e' \in E(H)} \gamma(e')$ for all $e \in E(H)$.
\end{proof}


On the other hand, it is easy to verify that if $x: V(H) \rightarrow [0,1]$ is a fractional $(\gamma,\varphi)$-balanced separator in $H$ with $\supp(x) \subseteq R$ and $y: E(H) \rightarrow [0,1]$ is a fractional edge cover of $x$, then by setting $x_v = x(v)$ for $v \in V(H)$, $y_e = y(e)$ for $e \in e(H)$, $d_{v,v'} = \mathsf{dist}^*_{H,x}(v,v')$ for $v,v' \in V(H)$, and $d_{e,e'} = \mathsf{dist}^*_{H,x}(e,e')$ for $e,e' \in V(H)$, we obtain a feasible solution of the LP.
This further implies that the LP must have a feasible solution.
Indeed, since $R$ is assumed to be a $(\gamma,\varphi)$-balanced separator, the function $x: V(H) \rightarrow [0,1]$ with $x(v) = 1$ for $v \in R$ and $x(v) = 0$ for $v \in V(H) \backslash R$ is a fractional $(\gamma,\varphi)$-balanced separator with $\supp(x) = R$.
The separator $x$ (together with a fractional edge cover of $x$) can be transformed to a feasible solution of the LP in the above way.

Now we are ready to prove Lemma~\ref{lem-balsepLP}.
We compute an optimal solution $\{\hat{x}_v,\hat{y}_e,\hat{d}_{v,v'},\hat{d}_{e,e'}\}$ of the LP, and return the function $x: V(H) \rightarrow [0,1]$ defined as $x(v) = \hat{x}_v$ for all $v \in V(H)$ as the fractional separator required in the lemma.
Observation~\ref{obs-distLB} guarantees that $x$ is truly a fractional $(\gamma,\varphi)$-balanced separator.
Because of the constraints $x_v = 0$ for all $v \in V(H) \backslash R$, we have $\supp(x) \subseteq R$.
To see that $\rho_H^*(x)$ is minimized, note that $\sum_{e \in e(H)} \hat{y}_e = \rho_H^*(x)$ by the optimality of the LP solution.
Consider another fractional $(\gamma,\varphi)$-balanced separator $x': V(H) \rightarrow [0,1]$ with $\supp(x') \subseteq R$ and a fractional edge cover $y': E(H) \rightarrow [0,1]$ of $x'$ with $\sum_{e \in E(H)} y'(e) = \rho_H^*(x')$.
We can transform $x'$ and $y'$ to a feasible solution of the LP in which the value of the objective function is equal to $\sum_{e \in E(H)} y'(e) = \rho_H^*(x')$.
The optimality of our LP solution implies $\sum_{e \in e(H)} \hat{y}_e \leq \sum_{e \in E(H)} y'(e)$, and equivalently, $\rho_H^*(x) \leq \rho_H^*(x')$.
Formulating and solving the LP can be done in $\lVert H \rVert^{O(1)}$ time.
This completes the proof of Lemma~\ref{lem-balsepLP}.

\subsection{Proof of Theorem~\ref{theorem:bal-separator-main}}

Based on the discussion in the previous sections, we can now prove our main theorem for computing minimum-cover balanced separator, i.e., Theorem~\ref{theorem:bal-separator-main}.
First, we use Lemma~\ref{lem-balsepLP} to compute a fractional $(\gamma,\varphi)$-balanced separator $x:V(H)\rightarrow [0,1]$ with $\supp(x) \subseteq R$ that minimizes $\rho_H^*(x)$, where $\varphi = \frac{1}{2}$.
Then we apply Theorem~\ref{thm-ballgrow} on $x$ to compute a $(\gamma,\varphi')$-balanced separator $S \subseteq \supp(x) \subseteq R$ in $H$, where $\varphi' = \frac{1+3\varphi}{2+2\varphi} = \frac{5}{6}$.
By Fact~\ref{fact-baldef}, $S$ is also a $(Z,\frac{5}{6})$-balanced separator in $H$.
Finally, as $\varphi = \frac{1}{2}$, Theorem~\ref{thm-ballgrow} implies
\begin{equation*}
    \rho_H^*(S) \leq (\min\{8+4\ln \alpha_H(R),6\mu_H\}+1) \cdot (104 + 16 \log \rho_H^*(x)) \cdot \rho_H^*(x).
\end{equation*}
Since $x$ is a fractional $(\gamma,\frac{1}{2})$-balanced separator with $\supp(x) \subseteq R$ that minimizes $\rho_H^*(x)$, by Fact~\ref{fact-intisfr}, $\rho_H^*(x) \leq \rho_H^*(\mathbf{1}_{S'}) = \rho_H^*(S')$ for any $(\gamma,\frac{1}{2})$-balanced separator $S' \subseteq R$.
Thus, $\rho_H^*(S)$ satisfies the desired bound in Theorem~\ref{theorem:bal-separator-main}.
The entire algorithm runs in $\lVert H \rVert^{O(1)}$ time, simply because both algorithms of Lemma~\ref{lem-balsepLP} and Theorem~\ref{thm-ballgrow} runs in $\lVert H \rVert^{O(1)}$ time.
\section{Algorithms for Fractional Hypertree Width}
\label{sec:fhtw-algorithm}
In this section we prove our two main results Theorem~\ref{theorem:main-poly-approx} and Theorem~\ref{theorem:main-FPT-approx}. We begin the section by showing the existence of a balanced separator with cover at most $\fhtw(H)$ for each set $Z\subseteq V(G)$. For this we use the notion of $(\gamma,\frac{1}{2})$-balanced separators that was introduced in Section~\ref{section:bal-sep}.

\begin{lemma}\label{lemma:2-balanced-sep-existence}
    Let $H$ be a hypergraph.
    For any function $\gamma: E(H) \rightarrow [0,1]$, there exists a $(\gamma,\frac{1}{2})$-separator $S \subseteq V(H)$ in $H$ such that $\rho_H^*(S) \leq \fhtw(H)$. 
\end{lemma}
\begin{proof}
Let $(T,\beta)$ be a hypertree decomposition of $H$ of fractional width $\fhtw(H)$.
We orient the edges of $T$ as follows.
Consider an edge $z \in E(T)$ connecting two nodes $t_1$ and $t_2$.
Denote by $T_1$ and $T_2$ the connected components of $T-z$ containing $t_1$ and $t_2$, respectively.
Define $V_1 = (\bigcup_{t \in T_1} \beta(t)) \backslash \beta(t_2)$ and $V_2 = (\bigcup_{t \in T_2} \beta(t)) \backslash \beta(t_1)$.
We orient $z$ towards $t_2$ if $\sum_{e \in E(H), e \cap V_1 \neq \emptyset} \gamma(e) \leq \sum_{e \in E(H), e \cap V_2 \neq \emptyset} \gamma(e)$, and orient $z$ towards $t$ otherwise.
We claim that $\sum_{e \in E(H), e \cap V_1 \neq \emptyset} \gamma(e) \leq \frac{1}{2} \sum_{e \in E(H)} \gamma(e)$ if $z$ is oriented towards $t_2$ and $\sum_{e \in E(H), e \cap V_2 \neq \emptyset} \gamma(e) \leq \frac{1}{2} \sum_{e \in E(H)} \gamma(e)$ if $z$ is oriented towards $t_1$.
The key observation is that no hyperedge of $H$ intersects both $V_1$ and $V_2$.
To see this, consider a hyperedge $e \in E(H)$.
We have $e \subseteq \beta(t)$ for some $t \in T$.
Assume $t$ and $t_1$ lie in the same connected component of $T-z$, without loss of generality.
If a vertex $v \in e$ is contained in $V_2$, then $v \in \beta(t_1)$ as the nodes whose bags containing $v$ must be connected in $T$.
But this contradicts the fact that $V_2 \cap \beta(t_1) = \emptyset$.
So we have $e \cap V_2 \neq \emptyset$ and thus no hyperedge intersects both $V_1$ and $V_2$.
It follows that
\begin{equation*}
    \sum_{e \in E(H), e \cap V_1 \neq \emptyset} \gamma(e) + \sum_{e \in E(H), e \cap V_2 \neq \emptyset} \gamma(e) \leq \sum_{e \in E(H)} \gamma(e).
\end{equation*}
Therefore, if $z$ is oriented towards $t_2$ (resp., $t_1$), then $\sum_{e \in E(H), e \cap V_1 \neq \emptyset} \gamma(e) \leq \frac{1}{2} \sum_{e \in E(H)} \gamma(e)$ (resp., $\sum_{e \in E(H), e \cap V_2 \neq \emptyset} \gamma(e) \leq \frac{1}{2} \sum_{e \in E(H)} \gamma(e)$).
  
Let $t^*$ be a sink of the resulting orientation of the edges of $T$.
We show that $\beta(t^*)$ is a $(\gamma,\frac{1}{2})$-balanced separator in $H$.
Suppose $T_1, \dots, T_r$ are the connected components of $T - t^*$, and define $V_i = (\bigcup_{t \in T_i} \beta(t)) \backslash \beta(t^*)$ for $i\in [r]$.
Since the edges incident to $t^*$ are all oriented towards $t^*$, as argued above, we have $\sum_{e \in E(H), e \cap V_i \neq \emptyset} \gamma(e) \leq \frac{1}{2} \sum_{e \in E(H)} \gamma(e)$ for all $i \in [r]$.
So it suffices to show that every connected component of $H - \beta(t^*)$ is contained in $V_i$ for some $i \in [r]$.
We first notice that $V_1,\dots,V_r$ are disjoint.
Indeed, if a vertex $v \in V(H)$ appears in bags of nodes in both $T_i$ and $T_j$ for different $i,j \in [r]$, then $v \in \beta(t^*)$.
This further implies that, for each hyperedge $e \in E(H)$, there exists at most one index $i \in [r]$ such that $e \cap V_i \neq \emptyset$.
To see this, note that $e \subseteq \beta(t)$ for some node $t \in T$.
If $t = t^*$, then $e \cap V_i = \emptyset$ for all $i \in [r]$.
Otherwise, $t \in T_i$ for some $i \in [r]$.
Then we have $e \subseteq V_i \cup \beta(t^*)$ and thus $e \cap V_j = \emptyset$ for all $j \in [r] \backslash \{i\}$.
Now consider a connected component $C$ of $H - \beta(t^*)$.
Clearly, $C \subseteq \bigcup_{i=1}^r V_i$.
Since each hyperedge of $H$ intersects at most one of $V_1,\dots,V_r$, $C$ must be contained in some $V_i$ for $i \in [r]$.
\end{proof}
    
Next in Lemma~\ref{lemma:Roberts-and-Seymours}, we show how to use the classic framework of Robertson and Seymour for treewidth approximation along with our algorithm for balanced separator (Theorem~\ref{theorem:bal-separator-main}) to prove Theorem~\ref{theorem:main-poly-approx} and Theorem~\ref{theorem:main-FPT-approx}. Lemma~\ref{lemma:Roberts-and-Seymours} is written in the form where in addition to $H$, the algorithm takes as input a tree decomposition $(\treeT,\beta)$ of $H$ and outputs another tree decomposition with better guarantees. We have it written this way so that we can use the algorithm as a black-box for both our \polyH\ time $\log(\alpha(H))$-approximation and $\norm{H}^{\width}$ time $f(\width)$-approximation. For the former, we will simply use the result with the input tree decomposition with all vertices in a single bag. For the latter, here is a quick summary of how we use the input tree decomposition.

Let $Z$ be the set of vertices for which we seek to find a balanced separator in the framework of Robertson and Seymour. We will use the given tree decomposition to guess a set of at most $\width$ bags whose union $R$ contains an optimum balanced separator of $Z$. By Corollary~\ref{corollary:cov_bags}, such a set exists. This guess incurs an additional $n^\width$ factor that is reflected in the running time but gives us an $\ln(\alpha_H(R))=\ln(\width\fhtw(\treeT))$ factor in the approximation as opposed to a factor of $\ln(\alpha_H)$. We use this gain in the approximation factor to design our $\norm{H}^\width$ algorithm. We also note that our algorithm might invoke Lemma~\ref{lemma:Roberts-and-Seymours} with $\treeT$ having $\fhtw(\treeT)$  much larger than $\width$. We are now ready to state and prove Lemma~\ref{lemma:Roberts-and-Seymours}.

\begin{lemma}\label{lemma:Roberts-and-Seymours}
Given a hypergraph $H$, a number $\width \geq 1$, and a tree decomposition $(\treeT,\beta)$ of $H$ with $p$ nodes and $\fhtw(\treeT)=\width_{\treeT}$, one can, in $p^{\lfloor\width\rfloor} \cdot \norm{H}^{O(1)}$  time, either construct a tree decomposition of $H$ with at most $|V(H)|$ nodes that has fractional hyper treewidth at most $c\cdot \min\{\ln(\width_{\treeT}\width),\ln\alpha_H,\mu_H\} \cdot \width \log\width$ for some constant $c>0$, or conclude that $\fhtw(H)>\width$.
\end{lemma}

\begin{proof}
We remark that the write-up for the framework is mostly adapted from \cite{cygan2015parameterized}.
We will assume that $H$ is connected, as otherwise we can apply the algorithm to each connected component of $H$ separately and connect the obtained tree decompositions arbitrarily.

Let $\lambda > 2\frac{\width'}{\varphi'}$ and $\lambda > 6\width$, where $\varphi' = \frac{1}{6}$ and $$\width'=(\min\{8+4\ln\width \width_{\treeT},8+4\ln\alpha(H),6\mu_H\}+1)\cdot(104+16\log\width)\cdot\width.$$ 
Define $\mathcal{R} = \{\bigcup_{t \in Y} \beta(t): Y \subseteq V(\treeT) \text{ with } 1 \leq |Y| \leq \width \}$.
In other words, $\mathcal{R}$ consists of the (nonempty) subsets of $V(H)$ which are the union of at most $\width$ bags in $(T,\beta)$.
%


Our algorithm uses a recursive procedure $\textsc{Decompose}(W,Z)$ for $Z\subset W\subseteq V(H)$. Here $Z$ will be the intersection of the subgraph $H[W]$ with the parent bag of the current partially computed decomposition. The procedure finds a tree decomposition for the subgraph $H[W]$ such that $Z$ is contained in the root bag.
During the recursive calls of this procedure, we want the following invariants:
\begin{enumerate}
    \item $\cov(Z)\leq \lambda$ and $W \backslash Z \neq \emptyset$.
    \item Both $H[W]$ and $H[W\setminus Z]$ are connected.
    \item $N_H(W\setminus Z)= Z$.
\end{enumerate}
The invariants capture the properties we want to maintain during the algorithm: at any point, we process some connected subgraph of $H$ which communicates with the rest of $H$ only via a small interface $H$.
The output of the procedure $\textsc{Decompose}(W,Z)$ is a (rooted) tree decomposition of the hypergraph $H[W]$ with fractional hyper treewidth at most $(1+\varphi')\lambda$ whose root bag contains $Z$.
The whole algorithm is then just to call $\textsc{Decompose}(V(H),\emptyset)$, which gives us a tree decomposition of $H$.

We now describe the procedure $\textsc{Decompose}(W,Z)$, it is also stated as Algorithm~\ref{alg-approx-fhtw} in pseudocode.
The first step is to construct the root bag $\hS$ with the following properties:
\begin{enumerate}
    \item[(a)] $Z\subset \hS\subseteq W$
    \item[(b)] $\cov(\hS)\leq (1+\varphi')\lambda$
    \item[(c)] For each connected component $C$ of $H[W] - \hS$, $\cov(N_H(V(C)) \cap \hS)\leq \lambda$.
\end{enumerate}

We now show how to construct $\hS$.
If $\cov(Z)\leq \lambda-1$, then we keeping adding vertices from $W\setminus Z$ to $Z$ until $\cov(Z) > \lambda-1$.
If at some point $W=Z$, we return the \textit{trivial} tree decomposition of $H[W]$ with a single node whose bag is $\hS=W$. 
Due to invariant $1$, $W\setminus Z\neq \emptyset$ at the beginning and since $\hS=W$, property $(a)$ is satisfied. Also $\cov(\hS)\leq \lambda$ and properties $(b)$, and $(c)$ trivially hold.

We now assume that $\lambda-1<\cov(Z)\leq \lambda$ and that $W\setminus Z\neq \emptyset$.
Let $\gamma:E(H[W])\rightarrow [0,1]$ be a fractional edge cover of $Z$ in $H[W]$ such that $\sum_{e \in E(H[W]) \gamma(e)} = \rho^*_{H[W]}(Z)$.
If $\fhtw(H)\leq \width$, then $\fhtw(H[W])\leq \width$.
Let $(\treeT,\beta_W)$ be the tree decomposition of $H[W]$ induced by $(\treeT,\beta)$, that is, $\beta_W(t) = \beta(t) \cap W$ for all $t \in T$.
Clearly, the fractional treewidth of $(\treeT,\beta_W)$ is at most $\width_{\treeT}$.
By Lemma~\ref{lemma:2-balanced-sep-existence}, there exists a $(\gamma,\frac{1}{2})$-balanced separator $S^*\subseteq V(H)$ with $\cov(S^*)\leq \width$.
Furthermore, by Corollary~\ref{corollary:cov_bags}, $S^*$ can be covered using at most $\lfloor\cov(S^*)\rfloor\leq \lfloor\width\rfloor$ bags in $(\treeT,\beta_W)$.
Let $R^*$ be the union of at most $\lfloor\width\rfloor$ bags in $(\treeT,\beta_W)$ such that $S^*\subseteq R^*$.
Applying Theorem~\ref{theorem:bal-separator-main} on $H[W],Z,R^*,\gamma$ yields a $(Z,1-\varphi')$-balanced separator $S\subseteq R^*$ with $\cov(S)\leq \width'$.
Observe that $\alpha_H(R^*)\leq \cov(R^*)\leq \min\{n,\width\width_\treeT\}$ and $R^*=R \cap W$ for some $R\in \mathcal{R}$.

In our $\textsc{Decompose}$ procedure, for each $R\in \mathcal{R}$ that is a $(\gamma,\frac{1}{2})$-balanced separator in $H$, we run Theorem~\ref{theorem:bal-separator-main} on $H[W],Z,R\cap W,\gamma$.
If in no iteration we find a $(Z,1-\varphi')$-balanced separator $S$ with $\cov(S)\leq \width'$ then we can safely conclude that $\fhtw(H[W])>\width$. Hence $\fhtw(H)>\width$ and the whole algorithm can be terminated.

Let us now assume that a $(Z,1-\varphi')$-balanced separator $S$ of $H[W]$ with $\cov(S)\leq \width'$ is obtained. We set $\hS=Z\cup S$. 
Initially when \textsc{Decompose} was called if $\cov(Z)<\lambda$ and some vertex from $W$ was added $Z$, then $\hS$ contains a vertex in $W\setminus Z$, satisfying property $(a)$. Suppose now for contradiction that $\hS \subseteq Z$, then we know by invariant that $H[W\setminus Z]$ is connected and $N_{H}(W\setminus Z)=Z$. Thus $H[W\setminus \hS]$ is connected. We know that $\cov((W\setminus \hS)\cap Z)\leq (1-\varphi')\cov(Z)\leq (1-\varphi')\lambda$. Further $\cov(\hS)\leq \omega'$. This implies $\cov(Z)\leq \omega' + (1-\varphi')\lambda \leq \frac{\varphi'}{2}\lambda+(1-\varphi')\lambda<\lambda$, which is a contradiction. Thus $Z\subset\hS\subseteq W$.

%
Also $\cov(\hS)\leq \cov(Z)+\cov(S)\leq \lambda + \width'\leq \lambda + \varphi' \lambda = (1+\varphi')\lambda$, satisfying property $(b)$. 

Recall that by definition of balanced separator, each connected component $C$ in $H[W]-S$ has $\cov(C\cap Z)\leq (1-\varphi')\cov(Z)$. To prove property $(c)$, let $C$ be a connected component of $H[W]-\hS$. Observe that $N_H(V(C))\cap \hS\subseteq S \cup (V(C') \cap Z)$, where $C'$ is the connected component in $H[W]-S$ that contains $C$. 
Thus $\cov(N_H(V(C))\cap \hS)\leq \cov(S)+\cov(C'\cap Z)\leq \width' + (1-\varphi')\cov(Z)\leq \width'+(1-\varphi')\lambda\leq \lambda$. The last inequality uses that $\lambda>\frac{\width'}{\varphi'}$.


With $\hS$ in hand, procedure $\textsc{Decompose}$ proceeds as follows. Let $C_1,\cdots,C_k$ be the connected components of $H[W]-\hS$ (possibly $k=0$).
For each $i\in [p]$, we then call recursively the procedure $\textsc{Decompose}(N_H[V(C_i)],N_H(V(C_i)))$.
Satisfaction of the invariants for these calls follows directly from the definition of $C_i$, invariant $(3)$ for the call $\textsc{Decompose}(W,Z)$ and properties $(a)$ and $(c)$ of $\hS$. 
Let $(\treeT_i,\beta_i)$ be the tree decomposition obtained from the call $\textsc{Decompose}(N_H[V(C_i)],N_H(V(C_i)))$, and let $r_i$ be its root; recall that $\beta_i(r_i)$ contains $N_H(V(C_i))$.
We now obtain a tree decomposition $(\treeT'_W,\beta'_W)$ of $H[W]$ by gluing the tree decompositions $(T_1,\beta_1),\dots,(T_k,\beta_k)$ together with $\hS$ as follows.
We create a root $r$ with bag $\beta'_W(r) = \hS$, and for each $i\in [k]$ attach the tree decomposition $(\treeT_i,\beta_i)$ below $r$ as a subtree.
This procedure is written as the sub-routine $\textsc{Glue}(\hat{S},(T_1,\beta_1),\dots,(T_k,\beta_k))$ in Algorithm~\ref{alg-approx-fhtw}, which returns us the tree decomposition $(\treeT'_W,\beta'_W)$.
From the construction it easily follows that $(\treeT'_W,\beta'_W)$ is a tree decompositon of $H[W]$ with root bag $\beta'_W(r) = \hS$. Since $\cov(\hS)\leq (1+\varphi')\lambda$, its fractional hypertree width is at most $(1+\varphi')\lambda$.

\begin{algorithm}
    \caption{\textsc{ApproxFhtw}$(H,\width,(\treeT,\beta))$}
    \begin{algorithmic}[1]
        \State $\mathcal{R} \leftarrow \{\bigcup_{t \in Y} \beta(t): Y \subseteq V(\treeT) \text{ with } 1 \leq |Y| \leq \width \}$ 
        \State \textbf{return} $\textsc{Decompose}(V(H),\emptyset)$
    \end{algorithmic}
    
 \textsc{Decompose}$(W,Z)$
    \begin{algorithmic}[1]
    \While{$\rho_H^*(Z) \leq \lambda-1$}
        \If{$Z=W$}{ \textbf{return} the trivial tree decomposition of $H[W]$}
        \Else{ $Z\leftarrow Z\cup \{w\}$ for an arbitrary $w \in W \backslash Z$}
        \EndIf
    \EndWhile

    \State $\hS\leftarrow \emptyset$ 
    \State $\gamma \leftarrow$ a fractional edge cover of $Z$ in $H[W]$ with $\sum_{e \in E(H[W])} \gamma(e) = \rho^*_{H[W]}(Z)$
    \For{every $R\in \mathcal{R}$ that is a $(\gamma,\frac{1}{2})$-balanced separator}
        \State $S \leftarrow \textsc{BalancedSeparator}(H[W],Z,R,\gamma)$
        \If{$S$ is a $(Z,\frac{5}{6})$-balanced separator with $\cov(S)\leq \width'$}
        \State $\hS\leftarrow Z\cup S$ and \textbf{break}
        \EndIf
    \EndFor

    \If{$\hS=\emptyset$}{ \textbf{return} \textsf{NO}}
    \EndIf
    
    \State $C_1,\cdots,C_k \leftarrow$ connected components of $H[W]- \hS$
    \For{$i = 1,\dots,k$}
        \If{$\textsc{Decompose}(N_H[V(C_i)],N_H(V(C_i)))$ returns \textsf{NO}}{ \textbf{return} \textsf{NO}}
            \EndIf  
      \State $(\treeT_{i},\beta_i)\leftarrow \textsc{Decompose}(N_H[V(C_i)],N_H(V(C_i)))$
        \State $r_i\leftarrow $ root of $(\treeT_i,\beta_i)$
    \EndFor
    
    \State \textbf{return} $\textsc{Glue}(\hat{S},(T_1,\beta_1),\dots,(T_k,\beta_k))$
    \end{algorithmic}
    
\label{alg-approx-fhtw}
\end{algorithm}

We now bound the runtime of our algorithm. By definition, $|\mathcal{R}|\leq p^{\lfloor\width\rfloor}$. Each call of $\textsc{Decompose}(W,Z)$ runs in time $|\mathcal{R}| \cdot \norm{H}^{O(1)}$. Also for each $R\in \mathcal{R}$, running Theorem~\ref{theorem:bal-separator-main} takes only \polyH\ time. 

Further each call of $\textsc{Decompose}(W,Z)$ adds exactly one new node to the final tree decomposition. Thus we bound the total number of nodes constructed by the algorithm to bound the number of calls. So we bound the number of nodes. Let $(\treeT^*,\beta^*)$ be the final tree decomposition. Each node in $\beta^*(t)$ is equal to $\hS$ for a $\hS$ obtained by some call $\textsc{Decompose}(W,Z)$. Here $\hS$ contains at least one vertex $v_t\notin Z$ by property $(a)$ and thus $\beta^*(t)$ is the topmost node in $\treeT^*$ of vertex $v_t$ by invariant $(3)$. This shows each node is a top most node for some vertex and thus $|V(\treeT^*)|\leq |V(H)|$. Thus the total runtime of the algorithm is $p^{\lfloor\width\rfloor} \cdot \norm{H}^{O(1)}$. 
This concludes the proof.
\end{proof}


We now show how to use Lemma~\ref{lemma:Roberts-and-Seymours} to prove Theorem~\ref{theorem:main-FPT-approx} 

\MainLgApprox*

\begin{proof}
Let $c$ be the constant from the width bound of Lemma~\ref{lemma:Roberts-and-Seymours}.
\begin{algorithm}[b]
    \caption{\textsc{FptFactorFhtw}$(H,\width)$}
    \begin{algorithmic}[1]
        \State $V(\treeT_0)\leftarrow\{r\}$, $E(\treeT_0)=\emptyset$, $B^0_r\leftarrow V(H)$
        \State $i\leftarrow 0$
        \While{$\width_i:=\fhtw(\treeT_i,\beta_i) > \lambda$}
            \State $i\leftarrow i+1$
            \State $(\treeT_i,\beta_i) \leftarrow \textsc{ApproxFhtw}(H,\width,(\treeT_{i-1},\beta_{i-1}))$
        \EndWhile
        \State \textbf{return} $(\treeT_i,\beta_i)$
    \end{algorithmic}
    \label{alg-fpt-fhtw}
\end{algorithm}
Our algorithm produces a sequence $(\treeT_i,\beta_i)$, $i\geq 0$, of hypertree decompositions of fractional hypertree width $\width_i$, until for some $i$ it holds that $\width_i\leq \lambda \width \log\width\ln\width$, where $\lambda=5c\ln 5c$. We start with a tree decomposition $(\treeT_0,\beta_0)$, consisting of a single bag containing all the vertices of $H$, here $\width_0 \leq \cov(H) \leq  n$. In each iteration we compute $(\treeT_i,\beta_i) = \textsc{ApproxFhtw}(H,\width,(\treeT_{i-1},\beta_{i-1}))$.
If the call to $\textsc{ApproxFhtw}$ returns $\textsc{NO}$, we return $\textsc{NO}$.
Refer to Algorithm~\ref{alg-fpt-fhtw} for a complete description of the algorithm. 

If the call to $\textsc{ApproxFhtw}$ in some iteration returns $\textsc{NO}$, then by correctness of $\textsc{ApproxFhtw}$ in Lemma~\ref{lemma:Roberts-and-Seymours}, $\fhtw(H)>\width$.
Thus, we assume that the algorithm outputs a decomposition. 
Further given $H,\width$ and a tree decomposition of $H$ of fractional hypertree width $\width_\treeT$, Lemma~\ref{lemma:Roberts-and-Seymours} either returns $\textsc{NO}$ or returns a tree decomposition of width at most $c\cdot\width \log\width\min\{\ln\width\width_{\treeT},\ln\alpha(H),\mu_H\}\leq c\cdot\width \log\width\min\{\ln\width\width_{\treeT},\ln n\}$ for some constant $c>0$. Thus $\width_1\leq c\cdot\width \log\width\ln n$ since $\alpha(H)\leq n$.

Now for each iteration $i\geq 1$, we show that $\width_i\leq 5c\cdot \width\log\width\ln^{(i)}n$ or $\width_i\leq 5c\ln(5c)\cdot \width\log\width\ln\width$. In the latter case, the algorithm terminates. We prove this by induction on $i$. For $i=0$, this directly follows since $\width_1\leq c\cdot\width \log\width\ln n$. We now assume $\width_j\leq 5c\cdot \width\log\width\ln^{(j)}n$ true for all $j<i$ and prove for $i$. We know that $w_i\leq c\cdot \width\log\width \ln(\width\width_{i-1})
\leq c\cdot \width\log\width \ln(\width \cdot 5c\cdot \width\log\width\ln^{(i-1)}n)
\leq 5c.\width\log\width \min\{\ln(5c)\ln\width,\ln^{(i)}n\}$. Due to the bound on $\width_i$, it is clear that the algorithm terminates after at most \polyH\ iterations with the guaranteed fractional hypertree width bound. The runtime follows from Lemma~\ref{lemma:Roberts-and-Seymours} and that we have \polyH iterations.
\end{proof}

Recall that $\mu_H$ is the degeneracy of the vertex-edge incidence graph of $H$ and $\eta_H$ is the maximum size of the intersection of any two hyperedges in $H$. We now show that $\mu_H\leq 2\eta_H\fhtw(H)$. This will help us obtain our desired bound on the approximation factor for Theorem~\ref{theorem:main-poly-approx}.

\begin{lemma}\label{lemma:intersection_bound}
For any hypergraph $H$, we have $\mu_H \leq 2 \eta_H \fhtw(H)$.
\end{lemma}
\begin{proof}
Let $\omega=\fhtw(H)$. Let $I$ be the vertex-edge incidence graph of $H$. It suffices to show that for every $E'\subseteq E(H)$ and $V' \subseteq V(H)$, the subgraph $I'$ of $I$ induced on $V'\cup E'$ has a vertex of degree at most $2\eta_H \omega$. Note that $I'$ here is the incidence graph of a hypergraph $H'$ with vertex set $V'$ and hyperedge set $\{e'\cap V':e'\in E'\} \setminus \{\emptyset\}$.

If $H'$ contains a hyperedge $e'$ of size at most $2\eta_H\omega$, then its corresponding vertex in $I'$ has degree at most $2\eta_H\omega$. It remains to consider the case where all hyperedges of $H'$ have sizes greater than $2 \eta_H\omega$. 
 
Let $H''\supseteq H'$ be the subgraph of $H$ restricted to the vertices of $V'$, i.e., $H''=(V', E'')$ where $E''=\{e'\cap V':e'\in E\} \setminus \{\emptyset\}$. Note that $\fhtw (H'') \leq \fhtw (H) = \omega$. Let $\mathcal T=(T,\beta)$ be a minimal hypertree decomposition of $H''$ of $\fhtw(T)\leq \omega$. By minimality, it has a leaf bag $\beta$ that contains some vertex $v\in V'$ along with all its neighbors in $H'')$. In particular, this means that the fractional edge cover number of the neighborhood $N(v)$ of $v$ in $H''$ is at most $\omega$. Let $\{y_e: e\in E''\}$ be the fractional edge cover of $N(v)$ such that $\sum_{e \in E''} y_e \leq \omega$.

We claim that $y_e\geq \frac{1}{2}$ for every hyperedge $e$ of $H''$ of size greater than $2\eta_H\omega$ containing $v$. Indeed, assume that there is a hyperedge $e_0\in E''$ of size $|e_0|>2\eta_H \omega$ such that $v\in e_0$ and $y_{e_0}<\frac{1}{2}$. Each vertex $u\in e_0 \subseteq N(v)$ should receive a total cover of at least one, so $\sum_{u\in e_0}\sum_{e \in E'': u \in e} y_e \geq |e_0|$. The hyperedge $e_0$ contributes less than $\frac{1}{2} \cdot |e_0|$ to this sum, so more than $\frac{1}{2} \cdot |e_0| \geq \eta_H \omega$ should be contributed by the rest of the hyperedges containing $v$. But any hyperedge $e\in E''\setminus \{e_0\}$ shares with $e_0$ at most $\eta_H$ vertices, so it contributes at most $\eta_H\cdot y_e$. In total, hyperedges other than $e_0$ contribute at most $\sum_{e\in E'': u\in e} \eta_H \cdot y_e \leq \eta_H \omega$, a contradiction. 

Recall that we are left with the case where all hyperedges of $H'$ have sizes larger than $2\eta_H\omega$. 
As we showed above, this means that $y'_e\geq \frac{1}{2}$ for every hyperedge $e$ of $H'\subseteq H''$ containing $v$. In particular, this means that $v$ belongs to at most $2 \omega$ hyperedges in $H'$, so it has degree at most $2 \omega$ in $I'$.
\end{proof}


We are now ready to prove Theorem~\ref{theorem:main-poly-approx}
\MainPolyTimeApprox*
\begin{proof}
Let  $(T,\beta)$ be a tree decomposition of $H$ with all vertices in a single bag. We run Lemma~\ref{lemma:Roberts-and-Seymours} on $H,\omega,$ and $(T,\beta)$.
First we observe that Lemma~\ref{lemma:Roberts-and-Seymours} runs in time \polyH\ because $T$ has only one bag. Next by Lemma~\ref{lemma:intersection_bound}, $\mu_H\leq \omega \eta_H$.
Thus if $\fhtw(H)\leq \omega$, then in time \polyH, Lemma~\ref{lemma:Roberts-and-Seymours} returns a tree decomposition of fractional hypertree width at most $c\width \log\width\cdot \min\{\ln(\width_{\treeT}\width),\ln\alpha_H,\mu_H\} \cdot$ which is $c\omega \log\omega\min\{\ln\alpha(H),\mu_H,\omega \cdot \eta_H\})$, for some $c>0$.
\end{proof}

Next we have a Lemma relating fractional and integral edge covers of a hypergraph $H$ with the $i$-$j$ intersection or bounded intersection property. Recall that this property guarantees that no set of $i$ distinct hyperedges in $H$ have more than $j$ vertices in common.

\begin{lemma}
\label{lem: bip_fractional_to_integral}
If $\cov(V(H)) \leq k$ and $H$ satisfies the i-j intersection property for some $i,j \in \mathbf{N}$ (that is, no subset of $i$ distinct hyperedges intersects in at least $j$ vertices) then $\rho_H(V(H))\leq 10 \cdot ki  \log (2kj)$. 
\end{lemma}
\begin{proof}

We proceed by induction on $k$. If $k=1$, then $H$ must contain a single hyperedge covering all the vertices, so it admits integral edge cover of size one. Assume that the statement holds for every $k\leq k_0$, we will show that it holds for every $k\in (k_0, k_0+\frac{1}{2i}]$. 

First, if $n \leq j \cdot (2k)^i$ then $\ln n \leq \ln j + i\cdot \ln (2k)$, so $H$ admits integral edge cover of size at most $k \log (1+\ln n) \leq k \cdot (1+\ln j + i \cdot \ln (2k))< 10 \cdot ki  \log (2kj)$ \cite{williamson2011design}. 

It remains to consider the case $n > j \cdot (2k)^i$. Let $y$ be the fractional edge cover of $H$ such that $\cost(y)=\sum_{e\in E} y_e \leq k \leq k_0+\frac{1}{2i}$.
We claim that $y_e \geq \frac{1}{2i}$ for at least one hyperedge $e$ of $H$. Assume the contrary, we will inductively construct a sequence of hyperedges $e_l$, $l \in [i]$, such that $|\cap_{t=1}^l e_t| \geq \frac{n}{(2k)^l}$.
First, since $\cov(H) \leq k$, there must be some $e_1$ such that $|e_1| \geq \frac{n}{k}$. Assume that the first $l<i$ hyperedges $e_t$, $t \in [l]$, are constructed. Let $S_l=\cap_{t=1}^l e_t$, then $|S_l|\geq \frac{n}{(2k)^l}$. Since $y_{e_t} \leq \frac {1}{2i}$ for all $t \in [l]$, the hyperedges $e_t$, $t \in [l]$, can contribute at most $\frac{l}{2i} \cdot |S_l|< \frac{1}{2}|S_l|$ to the total cover of vertices in $S_l$. So the remaining $\frac{1}{2}|S_l|$ of their total cover must be contributed by hyperedges
in $E_{l+1}=E\setminus \{e_t: t\in [l]\}$. 
There must be some hyperedge $e \in E_{l+1}$ such that $|e\cap S_l| > \frac{|S_l|}{2k}$, 
as otherwise we would have: $$\sum_{e\in E_{l+1}}y_e\cdot |e\cap S_l|< \frac{|S_l|}{2k} \sum_{e\in E_{l+1}} y_e \leq \frac{|S_l|}{2k} \cdot k = \frac{|S_l|}{2}.$$ 
Let $e \in E_{l+1}$ be any hyperedge such that $|e\cap S_l| > \frac{|S_l|}{2k}$, we can set $e_{l+1}=e$. After $i$ such iterations, we obtain $i$ hyperedges $e_t$, $t\in [i]$, such that  $|\cap_{t=1}^i e_t| \geq \frac{n}{(2k)^i}>j$, which contradicts to the fact that $H$ satisfies $i-j$ intersection property.

Hence, we conclude that $y_e \geq \frac{1}{2i}$ for some hyperedge $e$ of $H$. Then $\cov(H\setminus e)\leq k_0$, so by the induction hypothesis $H\setminus e$ admits integral edge cover of size at most $10 \cdot (k-\frac{1}{2i})i \log (2(k-\frac{1}{2i})j)$. By adding to it the hyperedge $e$, we obtain the integral edge cover of $H$ of size at most $10 \cdot (k-\frac{1}{2i})i \log (2(k-\frac{1}{2i})j)+1\leq 10 \cdot (ki-\frac{1}{2}) \log (2kj)+1 \leq 10 \cdot ki \log (2kj)$. 
\end{proof}
\section{Fractional $(A,B)$-Separators and Clique Menger's}
\label{sec:clique-menger}
In this section, we first show how to compute an approximate minimum-cover $(A,B)$ separator in \polyH\ time and prove our Clique Menger's Theorem (Theorem~\ref{thm:mainMengerSimplified}). Finally we also show that our rounding algorithm in Section~\ref{sec-ABrounding} has a gap of $O(\log|V(H)|)$ in the worst case.

\subsection{Computing minimum-cover $(A,B)$-separator} \label{sec-computingAB}
Here, based on Theorem~\ref{theorem:ab-sep-main}, we show how to design an approximation algorithm for computing an $(A,B)$-separator with minimum fractional edge cover number.
We remark that this algorithm is not used in any of our other results.

The algorithm is very simple: computing a \textit{fractional} $(A,B)$-separator $x: V(H) \rightarrow [0,1]$ with minimum fractional edge cover number and rounding $x$ to an $(A,B)$-separator $S \subseteq V(H)$ using Theorem~\ref{theorem:ab-sep-main}.
We show that $x$ can be computed by solving a linear program (LP).
Specifically, we prove the following lemma.

\begin{lemma}\label{lemma:lp-fractional-ab-sep}
Given a hypergraph $H$ and sets $A,B,R \subseteq V(H)$ such that $R$ is an $(A,B)$-separator in $H$, one can compute in $\lVert H \rVert^{O(1)}$ time a fractional $(A,B)$-separator $x:V(H)\rightarrow [0,1]$ with $\supp(x)\subseteq R$ that minimizes $\rho_H^*(x)$.
\end{lemma}

The variables of our LP are defined as follows.
For each vertex $v \in V(H)$, we introduce a variable $x_v \in [0,1]$ indicating the value of $x(v)$ for the desired separator $x$.
For each edge $e \in E(H)$, we introduce a variable $y_e \in [0,1]$; these variables together indicate a fractional edge cover $y: E(H) \rightarrow [0,1]$.
In order to guarantee that $x$ is a $(A,B)$-separator, we need additional variables to indicate the distances between vertices.
For every two vertices $v,v' \in V(H)$, we introduce a variable $d_{v,v'} \in [0,1]$ indicating the distance $\mathsf{dist}^*_{H,x}(v,v')$.
Our LP is formulated as
\begin{equation*}\tag{LP1}
    \begin{array}{rl}
        & \min \text{ } \sum_{e \in E(H)} y_e  \\[2ex]
        \text{s.t. } & \text{all variables $\in [0,1]$,} \\[1ex]
        & x_v = 0 \text{ for all $v \in V(H) \backslash R$,} \\[1ex]
        & \sum_{e \in E(H): v \in e} y_e \geq x_v \text{ for all $v \in V(H)$,} \\[1ex]
        & d_{a,b} \geq 1 \text{ for all $a \in A$ and $b \in B$,} \\[1ex]
        & d_{v,v} = x_v \text{ for all $v \in V(H)$,} \\[1ex]
        & d_{v,v''} \leq d_{v,v'} + x_{v''} \text{ for all $v \in V(H)$ and $(v',v'') \in E(\gf{H})$.}
    \end{array}
\end{equation*}

The constraints $x_v = 0$ for $v \in V(H) \backslash R$ guarantee $\supp(x)\subseteq R$, while the constraints $\sum_{e \in E(H): v \in e} y_e \geq 1$ for $e \in E(H)$ guarantee that the variables $y_e$'s represent a fractional edge cover of $x$.
Furthermore, the constraints $d_{a,b} \geq 1$ for $a \in A$ and $b \in B$ force $x$ to be a fractional $(A,B)$-separator.
The last two constraints are requirements that the distance function $\mathsf{dist}^*_{H,x}$ has to satisfy.
These constraints cannot guarantee that the distance variables are exactly equal to the values of $\mathsf{dist}^*_{H,x}$.
However, they are already sufficient for our purpose.
\begin{observation} \label{obs-distLBforAB}
    Let $\{\hat{x}_v,\hat{y}_e,\hat{d}_{v,v'}\}$ be a feasible solution of the LP.
    Define $x: V(H) \rightarrow [0,1]$ as $x(v) = \hat{x}_v$ for all $v \in V(H)$.
    Then $\hat{d}_{v,v'} \leq \mathsf{dist}^*_{H,x}(v,v')$ for all $v,v' \in V(H)$.
    In particular, $x$ is a fractional $(\gamma,\varphi)$-balanced separator in $H$.
\end{observation}
\begin{proof}
For each $v \in V(H)$, we have $\hat{d}_{v,v} = \hat{x}_v = \mathsf{dist}^*_{H,x}(v,v)$, because of the constraint $d_{v,v} = x_v$.
Assume $\hat{d}_{v,v'} > \mathsf{dist}^*_{H,x}(v,v')$ for some $v,v' \in V(H)$.
Consider a path $(v_0,v_1,\dots,v_\ell)$ in $H$ with $v_0 = v$ and $v_\ell = v'$ satisfying that $\sum_{j=0}^\ell x(v_j) = \mathsf{dist}^*_{H,x}(v,v')$.
Note that $\sum_{j=0}^i x(v_j) = \mathsf{dist}^*_{H,x}(v,v_i)$ for all $i \in \{0\} \cup [\ell]$, which implies $x(v_i) = \mathsf{dist}^*_{H,x}(v,v_i) - \mathsf{dist}^*_{H,x}(v,v_{i-1})$.
Let $k \in \{0\} \cup [\ell]$ be the smallest index such that $\hat{d}_{v,v_k} > \mathsf{dist}^*_{H,x}(v,v_k)$.
Clearly, such an index $k$ exists and $k > 0$, because $\hat{d}_{v,v_\ell} = \hat{d}_{v,v'} > \mathsf{dist}^*_{H,x}(v,v') > \mathsf{dist}^*_{H,x}(v,v_\ell)$ and $\hat{d}_{v,v_0} = d_{v,v} = \mathsf{dist}^*_{H,x}(v,v) = \mathsf{dist}^*_{H,x}(v,v_0)$.
Now 
\begin{equation*}
    \hat{d}_{v,v_k} - \hat{d}_{v,v_{k-1}} > \mathsf{dist}^*_{H,x}(v,v_k) - \mathsf{dist}^*_{H,x}(v,v_{k-1}) = x(v_k) = \hat{x}_k,
\end{equation*}
contradicting the constraint $d_{v,v_k} \leq d_{v,v_{k-1}} + x_{v_k}$.
Thus, $\hat{d}_{v,v'} \leq \mathsf{dist}^*_{H,x}(v,v')$ for all $v,v' \in V(H)$.

To see $x$ is a fractional $(A,B)$-separator, observe that $\hat{d}_{a,b} \geq 1$ for all $a \in A$ and $b \in B$, by the corresponding constraints of the LP.
As $\hat{d}_{a,b} \leq \mathsf{dist}^*_{H,x}(a,b)$ as argued above, we then have $\mathsf{dist}^*_{H,x}(a,b) \geq 1$ for all $e \in E(H)$.
\end{proof}

\begin{observation} \label{obs-septoLPsol}
    Let $x:V(H) \rightarrow [0,1]$ be a fractional $(A,B)$-separator with $\supp(x) \subseteq R$ and $y:E(H) \rightarrow [0,1]$ be a fractional edge cover of $x$.
    By setting $x_v = x(v)$ for $v \in V(H)$, $y_e = y(e)$ for $e \in E(H)$, and $d_{v,v'} = \mathsf{dist}^*_{H,x}(v,v')$ for $v,v' \in V(H)$, we obtain a feasible solution for LP1.
\end{observation}
\begin{proof}
One can easily verify from the definitions of fractional $(A,B)$-separators and fractional edge covers that the assignment satisfy all constraints in LP1.
\end{proof}

Now we are ready to prove Lemma~\ref{lemma:lp-fractional-ab-sep}.
We compute an optimal solution $\{\hat{x}_v,\hat{y}_e,\hat{d}_{v,v'}\}$ of the LP, and return the function $x: V(H) \rightarrow [0,1]$ defined as $x(v) = \hat{x}_v$ for all $v \in V(H)$ as the fractional separator required in the lemma.
Observation~\ref{obs-distLBforAB} guarantees that $x$ is truly a fractional $(A,B)$-separator.
Because of the constraints $x_v = 0$ for all $v \in V(H) \backslash R$, we have $\supp(x) \subseteq R$.
To see that $\rho_H^*(x)$ is minimized, note that $\sum_{e \in e(H)} \hat{y}_e = \rho_H^*(x)$ by the optimality of the LP solution.
Consider another fractional $(A,B)$-separator $x': V(H) \rightarrow [0,1]$ with $\supp(x') \subseteq R$ and a fractional edge cover $y': E(H) \rightarrow [0,1]$ of $x'$ with $\sum_{e \in E(H)} y'(e) = \rho_H^*(x')$.
By Observation~\ref{obs-septoLPsol}, we can transform $x'$ and $y'$ to a feasible solution of the LP in which the value of the objective function is equal to $\sum_{e \in E(H)} y'(e) = \rho_H^*(x')$.
The optimality of our LP solution implies $\sum_{e \in e(H)} \hat{y}_e \leq \sum_{e \in E(H)} y'(e)$, and equivalently, $\rho_H^*(x) \leq \rho_H^*(x')$.
Formulating and solving the LP can be done in $\lVert H \rVert^{O(1)}$ time.
This completes the proof of Lemma~\ref{lemma:lp-fractional-ab-sep}.

\begin{theorem}
    Given a hypergraph $H$ and sets $A,B,R \subseteq V(H)$ such that $R$ is an $(A,B)$-separator in $H$, one can compute in $\lVert H \rVert^{O(1)}$ time an $(A,B)$-separator $S \subseteq R$ in $H$ such that $\rho_H^*(S) \leq \min\{8+4\ln \alpha_H(R),6\mu_H\} \cdot \rho_H^*(S')$ for any $(A,B)$-separator $S' \subseteq R$ in $H$.
\end{theorem}
\begin{proof}
We first apply Lemma~\ref{lemma:lp-fractional-ab-sep} to compute a fractional $(A,B)$-separator $x: V(H) \rightarrow [0,1]$ in $H$ with $\supp(x) \subseteq R$ that minimizes $\rho_H^*(x)$.
We have $\rho_H^*(x) \leq \rho_H^*(S')$ for any $(A,B)$-separator $S' \subseteq R$ in $H$.
Then we apply Theorem~\ref{theorem:ab-sep-main} on $x$ to obtain an $(A,B)$-separator $S \subseteq R$ in $H$ such that $\rho_H^*(S) \leq \min\{8+4\ln \alpha_H(R),6\mu_H\} \cdot \rho_H^*(x)$.
The time complexity is clearly $\lVert H \rVert^{O(1)}$.
\end{proof}


\subsection{Clique Menger's Theorem} \label{sec-Mengers}
\mainMenger*

Let $\mathcal{P}$ be the set of $A$-$B$ paths in $G$. We now define an LP relaxation for the problem of finding an $A$-$B$ separator with minimum clique cover and also define its dual path packing LP.
\begin{equation*}
\tag{LP2}
\begin{array}{ll@{}ll}
\text{minimize}  & \displaystyle\sum\limits_{e\in \mathcal{F}} y_{e} &\\
\text{subject to}& \displaystyle\sum\limits_{e \in \mathcal{F}:e\cap V(P) \neq \emptyset}   &y_{e} \geq 1,  &\forall P\in \mathcal{P}\\
                 &                                                &y_{e} \in [0,1], &\forall e\in\mathcal{F}
\end{array}
\end{equation*}
\begin{equation*}
\tag{Dual}
\begin{array}{ll@{}ll}
\text{maximize}  & \displaystyle\sum\limits_{P\in \mathcal{P}} z_{P} &\\
\text{subject to}& \displaystyle\sum\limits_{P\in \mathcal{P}:e\cap V(P) \neq \emptyset}   &z_{P} \leq 1,  &\forall e\in \mathcal{F}\\
                 &                                                &z_{P} \in [0,1], &\forall P\in\mathcal{P}
\end{array}
\end{equation*}
Let $\textsf{opt}$ be the optimum cost of LP2 and its dual.
\begin{lemma}\label{lemma:menger-sep}
    If $\textsf{opt}\leq f$, then there is an $(A,B)$-separator $S$ of $G$ having $\rho^*_{G,\mathcal{F}}(S)\leq (8+4\ln n)f$.
\end{lemma}
\begin{proof}
    
Let $H$ be a hypergraph constructed from $G$ as follows: $V(H)=V(G)$, $E(H)=\mathcal{F}\cup E(G)$. Observe that $\gf{H}=G$. So each path in $G$ is a path in $\gf{H}$ and vice versa. Further any $S\subseteq V(H)$ is an $(A,B)$-seperator of $H$ if and only if $S$ is an $(A,B)$-separator of $G$ and $\rho^*_{H,\mathcal{F}}=\rho^*_{G,\mathcal{F}}(S)$. So to prove the Lemma, we will find an $(A,B)$-separator for $H$ with $\rho^*_{H,\mathcal{F}}\leq (8+4\ln n)f$.

Let $\hy=\{\hy_e:e\in \mathcal F\}$ be a feasible solution to LP2 of cost at most $f$. We will show that $x:V(H_G)\rightarrow [0,1]$, defined by $x(v)=\sum_{e \in \mathcal{F}:v \in e}\hy_e$, $v\in V(H_G)$, is a fractional $(A,B)$-separator in $H$ such that $\rho^*_{H,\mathcal{F}}(x)\leq f$.
Let $P\in \mathcal{P}$ be an $A$-$B$ path. Since $\hy$ is a feasible solution to LP2, it holds that $\sum_{e\in \mathcal{F}:e\cap V(P)\neq \emptyset}\hy_e\geq 1$.
Observe that $\sum_{v\in P}x(v)=\sum_{v\in P}\sum_{e\in \mathcal{F}: v\in e}\hy_e\geq \sum_{e\in \mathcal{F}:e\cap V(P)\neq \emptyset}\hy_e\geq 1$. Thus $\mathsf{dist}_{H,x}(A,B)\geq 1$. Further observe that by definition of $x$, $\hy$ is a fractional $\mathcal{F}$-edge cover of $x$. Therefore $x$ is a fractional $(A,B)$-separator with $\rho^*_{H,\mathcal{F}}(x)= \sum_{e\in \mathcal{F}}\hy_e\leq f$. 

By Theorem~\ref{theorem:ab-sep-stronger}, there exists an $(A,B)$-separator $S$ of $H$ satisfying $\rho_{H,\mathcal{F}}^*(S) \leq (8+4\ln n)\cdot \rho_{H,\mathcal{F}}^*(x)\leq (8+4\ln n)f$. This completes the proof.

\end{proof}

\begin{lemma}\label{lemma:menger-path}
    If $\textsf{opt}>f$, then there exists a collection $P_1, P_2, \ldots, P_{\lceil f \cdot \log(|{\cal F}|) \rceil}$ of $A$-$B$ paths in $G$ such that no clique in ${\cal F}$ intersects at least $6\log(|{\cal F}|)$ of the paths. 
\end{lemma}
\begin{proof}
  Let $\hz=\{\hz_P:P \in \cP\}$ be an optimum solution to the dual LP of cost $\sum_{P\in \cP}\hz_P = \textsf{opt}>f$. We define the function $z':\cP\rightarrow [0,1]$ by setting $z'(P):=\frac{\hz_P}{\sum_{Q\in \cP}\hz_Q}$ for each $P\in \cP$. Observe that $\sum_{P\in \cP}z'(P)=1$ and thus $z'$ is a probability distribution on $\cP$. Let $t=\log{|\cF|}$ and $\ell=\lfloor ft \rfloor$, we sample independently one by one $\ell$ paths $P_1, P_2, \ldots, P_\ell$  from the distribution $z'$ (possibly with repetitions). 
  For each $e\in \mathcal{F}$ and $i\in [\ell]$, the following holds $$Pr[P_i\cap e\neq\emptyset]= \sum_{P \in \cP:e\cap V(P)\neq \emptyset}z'(P)=\sum_{P \in \cP: e\cap V(P)\neq \emptyset}\frac{\hz_P}{\sum_{Q\in \cP}\hz_Q}= \frac{1}{\sum_{Q\in \cP}\hz_Q}\cdot\sum_{P \in \cP:e\cap V(P)\neq \emptyset}\hz_P< \frac{1}{f},$$
where the last inequality follows from the fact that $\hz=\{\hz_P:P \in \cP\}$ is an optimum solution to the dual LP and it has cost greater than $f$. Let $X_{e,i}$ be an indicator random variable that denotes whether $P_i$ intersects $e\in \cF$ or not. Further let $X_e=\sum_{i\in[\ell]}X_{e,i}$  denote the number of paths in $\{P_1,\cdots,P_\ell\}$ that intersect $e$. 
$$E[X_e]=\sum_{i\in[\ell]}E[X_{e,i}]= \sum_{i\in[\ell]} Pr[P_i\cap e\neq\emptyset] \leq \sum_{i\in[\ell]} \frac{1}{f} \leq \frac{l}{f}\leq \frac{ft}{f}\leq t.$$

We use the following variant of Chernoff bounds, $Pr[X_e\geq R]\leq 2^{-R}$ for every $R\geq 6 \cdot E[X_e]$ since $X_e$ is the sum of independent random variables. By applying this to $R= 6 t = 6 \log{|\cF|}$ we obtain:
$$Pr[X_e\geq 6t]= Pr[X_e\geq 6 \log{|\cF|}]\leq 2^{-6\log{|\cF|}}\leq \frac{1}{|\cF|^{6}},$$
so $$Pr[\exists e\in \mathcal{F}, X_e\geq 6t ]\leq \frac{1}{|\cF|^5}$$ 
$$Pr[\forall e\in \mathcal{F}, X_e<6t ]\geq 1- \frac{1}{|\cF|^5}$$ 

Since $Pr[\forall e\in \mathcal{F}, X_e<6t ]>0$, there exists a collection of $\lceil ft \rceil$ paths in $\cP$ such that no clique in $\cF$ intersects $\geq 6t$ of the paths in the collection.
\end{proof}
Lemma~\ref{lemma:menger-sep} and Lemma~\ref{lemma:menger-path} together prove our Clique Menger's Theorem, Theorem~\ref{thm:mainMengerSimplified}.
\subsection{Logarithmic gap between fractional and integral $(A,B)$-separators}
Our rounding algorithm in Section~\ref{sec-ABrounding} has a gap of $O(\log |V(H)|)$ in the worst case.
In this section, we show this bound is tight.
Specifically, we give a family of hypergraphs in which there exist fractional $(A,B)$-separators with constant fractional edge cover number, while every $(A,B)$-separator must have fractional edge cover number $\Omega(\log |V(H)|)$. 

\begin{theorem}\label{thm:gapIsTight}
For every $n \geq 3$, there exists a hypergraph $\mathcal I_n$ on $N=\bigoh (2^n)$ vertices and a pair of its vertices $a,b$ such that:
\begin{itemize}
    \item $\mathcal I_n$ admits a fractional $(a,b)$-separator $x$ with $\cov(x)\leq 8$.
    \item for every integral $(a,b)$-separator $S$ it holds that $\cov(S) = \Omega(\log(N))$.  
\end{itemize}
\end{theorem}
\begin{proof}
Let $C$ be the unit circle $[0,1]$ with $0 \sim 1$. We define the interval graph $\mathcal I_n$ with intervals $\{I_{j,0}^k, I_{j,1}^k: j \in[n], k \in [2^{j+1}-1]_0\} \subseteq C$, where $I_{j,b}^k=(\frac{k-1}{2^{j+1}}, \frac{k+1}{2^{j+1}})$ for all $j \in [n]$,
$k \in [2^{j+1}-1]_0$ and $b\in \{0,1\}$, see Figure \ref{fig: interval_graph} for the illustration. This means that $\mathcal I_n$ contains a vertex for each interval $I_{j,b}^k$, and an edge (hyperedge of size 2) between every pair of vertices whose intervals intersect. Let $E_{\operatorname{short}}$ be the set of all such edges. In addition,  $\mathcal I_n$ contains hyperedges of size $n$ that we will define further.  

\begin{figure}
\label{fig: interval_graph}
\begin{center}\includegraphics[width=0.7\textwidth]{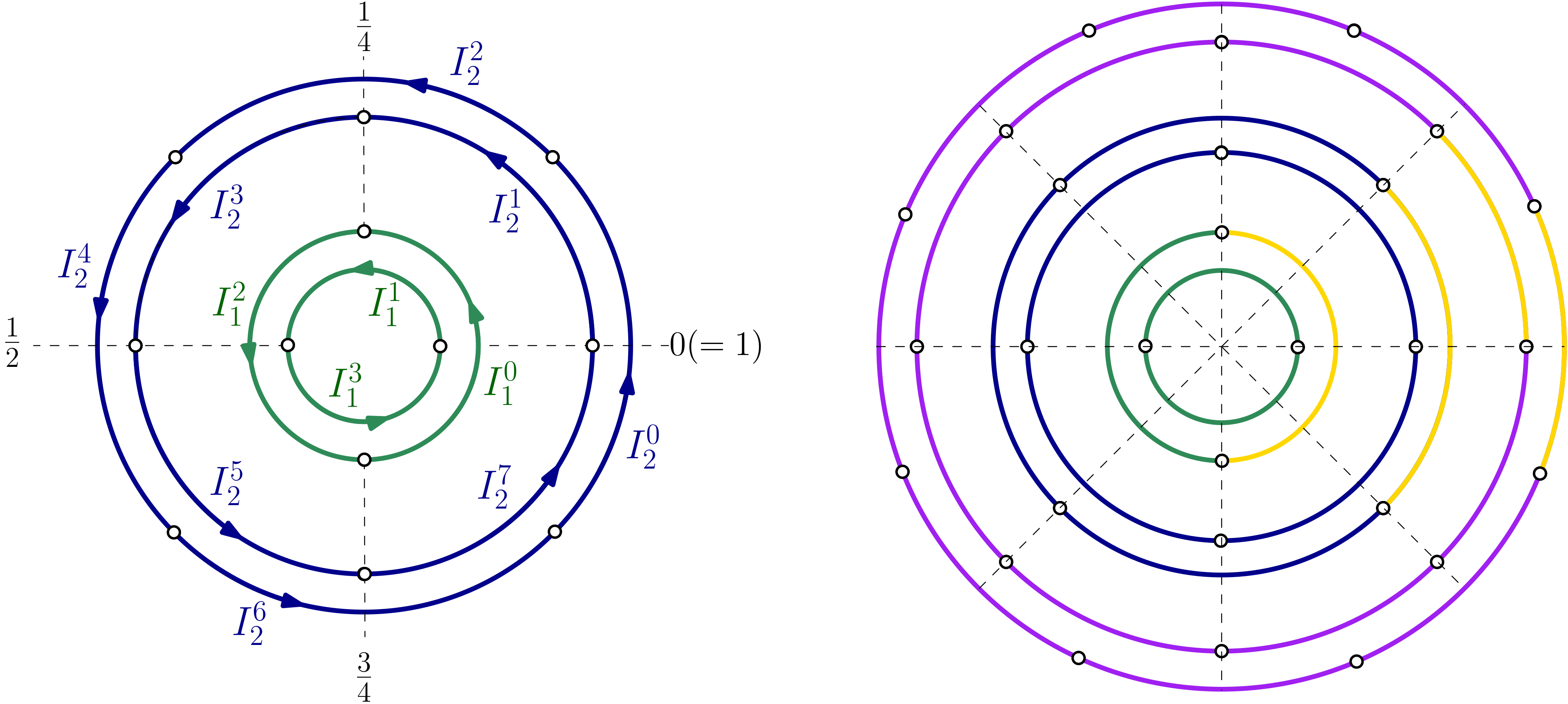}\end{center}
\caption{Interval graphs $\mathcal I_2$ (left) and $\mathcal I_3$ (right), where we depict each pair of intervals $I_{j,0}^k$ and $I_{j,1}^k$ by a single interval $I_{j}^k$ for simplicity. In a green layer, we depict intervals $I_{1}^k$, in a blue layer $I_{2}^k$, and in a purple layer $I_{3}^k$, for all values of $k$. In $\mathcal I_3$, we highlight the spanning tree rooted in $I_{1,0}^0$ by yellow. The root has children $I_{2,0}^0$ and $I_{2,1}^0$ (both depicted as $I_{2}^0$ in the blue layer). The leaves are the children of $I_{2,0}^0$ ($I_{3,0}^0$ and $I_{3,1}^0$) and the children of $I_{2,1}^0$ ($I_{3,0}^1$ and $I_{3,1}^1$), see the purple layer.}
\end{figure}

For simplicity, we identify vertices with their intervals. Let $a$ and $b$ be the opposite intervals $I^0_{n,0}$ and $I^{2^n}_{n,0}$ respectively. To describe an optimal fractional $(a,b)$-separator $x$, we first define a spanning forest $\mathcal F$ of $\mathcal I_n$ as follows. The root vertices are $I_{1,b}^k$, $0 \leq k \leq 3$, $b\in \{0,1\}$. For every $j \in [n]\setminus \{1\}$, $k \in [2^{j+1}-1]_0$ and $b\in \{0,1\}$ the vertex $I_{j,b}^k$ has parent $I_{j-1, k \mod 2}^{\lfloor \frac {k}{2} \rfloor}$. Note that every non-leaf vertex has precisely two children with equal intervals in $\mathcal F$, and the child is always a subinterval of its parent. Therefore, any neighbor of a child in $\mathcal I_n$ is a neighbor of its parent. In particular, the leaf-to-root paths of $\mathcal F$ form cliques in $\mathcal I_n$. For every leaf $t$ of $\mathcal F$, we define a hyperedge containing $t$ along with all its ancestors in $\mathcal F$. Let $E_{\operatorname{long}}$ be the set of all such hyperedges, we then set $E(\mathcal I_n)=E_{\operatorname{short}} \cup E_{\operatorname{long}}$, which finalizes the construction of $\mathcal I_n$.

We obtain the fractional fractional $(a,b)$-separator $x$ along with its fractional edge cover $y$ as follows. For every hyperedge $e \in E_{\operatorname{long}}$ that does not contain $I^0_{n,0}$ and $I^{2^{n}}_{n,0}$, we set $y_e=\frac{4}{2^{n+1}-1}$. For any other hyperedge $e$ of $\mathcal I_n$, we set $y_e=0$. Furthermore, for every vertex $v$ of $\mathcal I_n$, we set $x_v=\sum_{e:v\in e} y_e$. Since $\mathcal F$ has $2^{n+2}$ leaves and every leaf $l$ other than $I^0_{n,0}$ and  $I^{2^{n}}_{n,0}$ uniquely identifies a hyperedge $e$ such that $l\in e$ and $y_e=\frac{4}{2^{n+1}-1}$, this results in $\cost(y)=8$. To see that $x$ is a fractional $(a,b)$-separator, consider any path $p$ between $a=I^0_{n,0}$ and $b=I^{2^{n}}_{n,0}$. Let $(I_{j_i, b_i}^{k_i} : i \in [q])$ be the sequence of vertices of $p$ excluding its endpoints. By construction of $\mathcal I_n$, any interval intersecting $I^0_{n,0}$ contains the center of $I^0_{n,0}$, the same holds for $I^{2^{n}}_{n,0}$. The distance between the centers of $I^0_{n,0}$ and $I^{2^{n}}_{n,0}$ along the unit circle is $\frac{1}{2}$, which lower-bounds the total length of the internal intervals of $p$, i.e. $\sum_{i=1}^q \frac {1}{2^{j_i}} \geq \frac {1}{2}$. By our construction, a vertex $v_i=I_{j_i,b_i}^{k_i}$ belongs to $2^{n-j_i}$ leaf-to-root paths in the spanning forest, so $x_{v_i}=\frac{2^{n-j_i+2}}{2^{n+1}-1}>\frac{2^{n-j_i+2}}{2^{n+1}}=\frac {2}{2^{j_i}}$. Then the sum of $x_v$ over intermediate vertices $v$ of $p$ is at least $\sum_{i=1}^p \frac {2}{2^{j_i}} \geq 2\cdot \frac {1}{2}=1$. Therefore, $x$ is a fractional $(a,b)$-separator. By construction, $y$ is fractional edge cover of $x$, recall that $\cost(y)=8$.

Let $S$ be any integral vertex separator between $I^0_{n,0}$ and $I^{2^{n}}_{n,0}$. First, observe that if some non-root vertex belongs to $S$, so does its parent. Indeed, the neighborhood of a parent in $\mathcal I_n$ contains the neighborhood of its child. Furthermore, any pair of vertices $I_{j,0}^k$ and $I_{j,1}^k$ have precisely the same neighborhoods, so if one of them belongs to $S$, so does the other. By these two observations, we conclude that for each $j\in [n]$, $S$ contains at least two vertices $I_{j,0}^{k_j}$ and $I_{j,1}^{k_j}$ that are children of the same vertex of $S$ if $j>1$. Hence, $S$ contains subset $S'$ that consists of $n$ vertices of the same leaf-to-root path in $\mathcal F$ (we will denote this path by $p_0$) plus at least $n$ vertices such that no two of them belong to the same leaf-to-root path, we will call such vertices \emph{incomparable}. Therefore, $S'$ contains $n$ vertices of $p_0$ and $n$ incomparable vertices, resulting in $|S'|\geq 2n$.

Let $y'$ be some fractional edge cover of $S$, then in particular it is fractional edge cover of $S'$. 
Let $\tau=\sum_{e\in E(\mathcal I_n)} y'_e \cdot |e\cap S'|$, then $\tau \geq |S'| \geq 2n$.
At the same time, each $e\in E_{\operatorname{short}}$ contributes at most $2 y'_e$ to $\tau$. Hence, if hyperedges from $E_{\operatorname{short}}$ contribute at least $\frac{n}{2}$ in total, then $\sum_{e\in E_{\operatorname{short}}} y_e' \geq \frac{n}{4}$, so in particular $\cost(y') \geq \frac{n}{4}$. Otherwise, since hyperedge $e_0$ defined by the leaf-to-root path $p_0$ contributes $n\cdot y'_{e_0} \leq n$ to $\tau$, rest of the hyperedges $e\in E_{\operatorname{long}}$ must contribute at least $|S'|-n-\frac{n}{2} \geq \frac{n}{2}$. Since $|e\cap S'| \leq 1$ for every $e\in E_{\operatorname{long}}$, their contribution is upper-bounded by sum of their costs $\sum_{e\in E_{\operatorname{long}}} y_e'\leq \cost (y')$, so it must hold that $\cost(y') \geq \frac{n}{2}$. Therefore, in either case $\cost(y') \geq \frac{n}{4} =\bigoh (\log N)$.

We note that in this example the fractional hypertree width is logarithmic in the number of vertices as well. Indeed, consider any tree decomposition $\mathcal T$ of $\mathcal I_n$. Let $S_0$ be a set of vertices of $\mathcal I_n$ consisting of some leaf-to-root path of $\mathcal F$ plus the siblings of its non-root vertices. As we just showed, such a set contains $n$ incomparable vertices. Every hyperedge of $\mathcal I_n$ covers at most two incomparable vertices of $S_0$, so $S_0$ can not be covered by less that $\frac {n}{2}$ hyperedges. Intervals of any two siblings coincide, so $S_0$ forms a clique in $\mathcal I_n$ and hence it is fully contained in some bag of $\mathcal T$. Therefore, the fractional hypertree width of $\mathcal T$ is at least $\frac {n}{2}$. Since this holds for any tree decomposition $\mathcal T$ of $\mathcal I_n$, we conclude that $\fhtw(\mathcal I_n)\geq \frac {n}{2}$
\end{proof}


\section{Conclusion and Open Problems}
\label{sec:conclusion}
We gave two new approximation algorithms to compute the fractional hypertree width of an input hypergraph $H$. In the light of our work it is a very intriguing open problem whether it is possible to attain a truly polynomial time $\cO(\log^{O(1)} \omega)$-approximation algorithm for fractional hypertree width without any exponential running time dependence on $\omega$.
A less ambitious but still very intersting goal is to obtain an FPT-approximation algorithm for fractional hypertree width, that is an $f(\omega)$-approximation algorithm running in time $g(\omega)|H|^{\cO(1)}$ for some functions $f$ and $g$.
The main hurdle to obtaining such an approximation algorithm using our methods is that the $\cO(\log n)$ gap in the rounding scheme of Theorem~\ref{theorem:ab-sep-main} is tight (see Theorem~\ref{thm:gapIsTight}). 

On the way to obtaining our algorithms for fractional hypertree width we proved a variant of Menger's Theorem for clique separators in graphs. This theorem has already found an application in structural graph theory~\cite{ChudGHLS24}. We hope and anticipate to see further applications of the rounding scheme for clique separators and for the Clique Menger's theorem.

\bibliographystyle{alpha}
\bibliography{hyptw}
\end{document}